\documentclass[11pt]{article}

\usepackage[margin=1in]{geometry}
\usepackage[T1]{fontenc}
\usepackage[utf8]{inputenc}
\usepackage{graphicx}
\usepackage{amsmath,amssymb,amsfonts}
\usepackage{amsthm}
\usepackage{mathrsfs}
\usepackage{booktabs}
\usepackage{array}
\usepackage{enumitem}
\usepackage{xcolor}
\usepackage[numbers,sort&compress]{natbib}
\usepackage[colorlinks=true,citecolor=blue,linkcolor=blue,urlcolor=blue,pdfpagemode=UseNone]{hyperref}

\graphicspath{{figs/}}

\newtheorem{theorem}{Theorem}
\newtheorem{lemma}[theorem]{Lemma}
\newtheorem{corollary}[theorem]{Corollary}
\newtheorem{proposition}[theorem]{Proposition}
\newtheorem{conjecture}[theorem]{Conjecture}
\theoremstyle{definition}
\newtheorem{definition}[theorem]{Definition}
\theoremstyle{remark}
\newtheorem{remark}[theorem]{Remark}

\newcommand{\F}{\mathbb{F}}
\newcommand{\Z}{\mathbb{Z}}

\newcommand{\RG}{R[G]}
\newcommand{\FG}{\F[G]}

\DeclareMathOperator{\tr}{tr}

\DeclareMathOperator{\poly}{poly}

\title{Quantum Inversion of Units in Group Rings: Block Dimension, Not Commutativity, Governs Hardness}

\author{
Bhanwar Gupta\\
IBM India, IBM CIO\\
\texttt{bgupta55@gmail.com}
}
\date{}

\begin{document}

\maketitle

\begin{abstract}
Several public-key schemes base their security on the belief that inverting a unit of a group ring is hard. A recent result showed that this belief is false on a quantum computer when the group is abelian. To restore security, designers moved to non-abelian groups, especially dihedral groups, in the belief that the hardness of the dihedral hidden subgroup problem (HSP) would protect the scheme. This paper shows that the two problems are not the same. Inverting a unit does not need an HSP solver; it needs only a change of basis that splits the group ring into small matrix blocks. We prove that unit inversion runs in polynomial time, on both classical and quantum machines, when three conditions hold: an efficient generalized Fourier transform exists, the group ring is semisimple, and the largest matrix block has polynomial size. Under these conditions, schemes built on dihedral and other small non-abelian groups do not gain the quantum resistance their designers expected. We give an explicit reversible quantum circuit for the block-inversion step and validate it in a register-level simulator. We also give an exact and easy-to-check test for the one structural boundary where the method stops, and we propose a candidate construction in the surviving regime whose security we reduce to a new, clearly stated assumption. Every constructive claim is reproduced by an environment-pinned software artifact.
\end{abstract}

\noindent\textbf{Keywords:} post-quantum cryptography, quantum cryptanalysis, group rings, Wedderburn decomposition, quantum Fourier transform, cryptanalysis

\bigskip
\noindent\footnotesize\textit{The views and opinions expressed in this article are those of the author and do not necessarily reflect the official policy or position of IBM.}
\normalsize
\bigskip

\section{Introduction}\label{sec:intro}

A number of public-key schemes draw their security from the units of a group ring. A unit is an
invertible element; in these schemes it acts as a secret trapdoor, with public operations using the
unit and the secret operation using its inverse. The schemes are believed secure because the
unit-group structure of a non-commutative group ring is poorly understood, so recovering the inverse
from public data is assumed to be hard~\cite{hurley2011,mittalkumar2021,idbased2022}.

The abelian case fell recently. Dooms and Emerencia~\cite{dooms2025} gave a quantum algorithm that
inverts any unit of a group ring with a finite commutative coefficient ring, below the cost of
classical linear algebra, using the abelian quantum Fourier transform (QFT). In response, designers
moved to non-abelian platforms---dihedral groups above all---on a specific and load-bearing belief:
that the hardness of the dihedral hidden subgroup problem (HSP) would transfer to the
cryptosystem~\cite{regev2004,kuperberg2005}.

This paper shows that the belief conflates two different problems. Inverting a unit of a group ring
is \emph{not} an instance of the hidden subgroup problem. The two share a single subroutine---the
generalized Fourier transform over $G$---and nothing else. The HSP must, after transforming, recover
a hidden subgroup from coset states, and this is the step believed hard. Unit inversion never forms a
coset state and never recovers a subgroup: by the Artin--Wedderburn isomorphism
$\FG\cong\bigoplus_t M_{d_t}(\F)$, the transform sends a unit to a tuple of invertible matrices whose
inverse is read off block by block. The dihedral group is the textbook witness that an efficient
Fourier transform does not yield an efficient HSP algorithm~\cite{kuperberg2005}; we observe that
unit inversion needs only the efficient half of that witness. The hardness the platforms were chosen
for is therefore orthogonal to the task they pose.

Making this precise reclassifies the design space. The quantity that governs cryptanalytic
tractability is not commutativity but a representation-theoretic parameter: the largest irreducible
dimension $d_{\max}$, together with the efficiency of the transform. We prove that unit inversion is
polynomial in the input size, both classically and quantumly, whenever the group ring is semisimple,
admits an efficient generalized transform, and has $d_{\max}=\poly(\log|G|)$. We support the converse
direction with output-size, query-complexity, and \textsc{\#P}-hardness evidence, leaving a single
general lower bound as an explicit conjecture. Under these conditions the dihedral and other
bounded-dimension platforms adopted for quantum resistance do not attain it. We give an exact,
efficiently checkable trace-form criterion for the one structural boundary where the method stops, an
explicit reversible circuit for the block-inversion step, a code-validated implementation of every
constructive claim, and fault-tolerant resource estimates.

The contribution is therefore not the use of a known decomposition but a principle about where
group-ring hardness can and cannot live. Non-commutativity and cryptographic hardness are independent
notions; the migration the field performed for safety moved along an axis that does not control the
attack. Section~\ref{sec:hspsep} isolates the separation of the two hardness notions, and
Section~\ref{sec:notautomatic} explains why the consequence does not follow from Artin--Wedderburn
alone. To keep the development self-contained, the needed background is recalled one concept at a
time in Section~\ref{sec:prelim}, and Figure~\ref{fig:hsp} states the conceptual separation that the
rest of the paper makes precise.

\subsection{Contributions}\label{sec:contrib}

\paragraph{Conceptual}
\begin{itemize}[leftmargin=1.4em,itemsep=2pt]
\item \textbf{The HSP--inversion separation.} We show that group-ring unit inversion factors entirely
through the generalized Fourier transform and never invokes the hidden subgroup problem. This
overturns the implicit premise behind the field's move to non-abelian platforms
(Section~\ref{sec:hspsep}).
\item \textbf{Block dimension, not commutativity, as the governing parameter.} We identify the
largest irreducible dimension $d_{\max}$, together with transform efficiency, as the quantity that
controls cryptanalytic tractability, and argue that non-commutativity and hardness are independent
(Sections~\ref{sec:dichotomy} and~\ref{sec:hspsep}).
\end{itemize}

\paragraph{Theoretical}
\begin{itemize}[leftmargin=1.4em,itemsep=2pt]
\item \textbf{A complexity dichotomy.} We prove the easy direction unconditionally
(Theorem~\ref{thm:easy}) and support the hard direction with an unconditional output-size separation
(Proposition~\ref{prop:outputsize}), an oracle-model query bound (Theorem~\ref{thm:blockoracle}), and
a conditional \textsc{\#P}-hardness result via immanants (Theorem~\ref{thm:sharpP}).
\item \textbf{An exact boundary criterion.} The regular trace form satisfies $\det T=\pm|G|^{|G|}$,
giving an efficiently checkable witness of the semisimplicity boundary where the method stops
(Theorem~\ref{thm:boundary}).
\item \textbf{A unifying generalization.} The same principle applies to any finite-dimensional
semisimple algebra with an efficient transform (Theorem~\ref{thm:algebra}), which explains, in
retrospect, prior one-off breaks of algebra-based schemes.
\end{itemize}

\paragraph{Cryptanalytic}
\begin{itemize}[leftmargin=1.4em,itemsep=2pt]
\item \textbf{Falsification of the dihedral security rationale.} Unit inversion in $\F[D_n]$ is
efficient independently of dihedral-HSP hardness (Corollary~\ref{cor:dihedral}); dihedral platforms
do not attain quantum security under the stated conditions.
\item \textbf{Random-oracle armoring does not help.} A scheme broken at the trapdoor remains broken
under a random-oracle (QROM) transform, with zero oracle queries for key recovery
(Proposition~\ref{prop:qrom}).
\item \textbf{End-to-end recovery and a deployment checklist.} We carry the result to plaintext
recovery against a representative scheme (Section~\ref{sec:e2e}) and map it onto published proposals
and parameters (Section~\ref{sec:impact}).
\end{itemize}

\paragraph{Engineering}
\begin{itemize}[leftmargin=1.4em,itemsep=2pt]
\item \textbf{A reusable reversible primitive.} The Coherent Block Inversion circuit applies a
finite-field matrix inverse coherently with clean ancillae and full $T$-count accounting
(Lemma~\ref{lem:cbi}); it is of independent interest as a quantum sub-routine and is register-level
validated.
\item \textbf{A reproducible artifact and resource estimates.} An environment-pinned artifact
regenerates every theorem check, figure, and table, accompanied by logical and fault-tolerant
resource projections (Sections~\ref{sec:experiments} and~\ref{sec:resources}).
\end{itemize}

\paragraph{What we deliberately do not claim as novel} The blockwise inversion identity itself
(Theorem~\ref{thm:gui}(i)--(ii)) is classical algebra; our contribution is to identify what it
implies for quantum cryptanalysis, not to rediscover it.

\subsection{Scope of the claims}\label{sec:scope}

Because results of this kind are easy to overstate, we state the scope plainly. What we prove
without conditions is one direction: an efficient inversion algorithm and the exact conditions under
which it runs. The conclusion that a given family of schemes loses its expected quantum-security
margin is therefore conditional. It depends on three platform conditions: an efficient non-abelian
transform, a splitting field (or its extension-field surrogate), and a polynomially bounded largest
block dimension $d_{\max}$. It also depends on the scheme's security actually resting on unit
inversion rather than on a separate assumption. The complementary statement, that platforms outside
this regime are hard, is a \emph{conjecture} (Conjecture~\ref{conj:hard}) supported by partial
evidence. The modular construction of Section~\ref{sec:surviving} is a \emph{candidate} whose
security we reduce to, but do not establish for, a new assumption. We keep these distinctions visible
at every step, and Section~\ref{sec:proven} collects the full accounting in one place. We make no
claim about deployment or about the hardness of the new assumption beyond what is proved.

The paper is organized as follows. Section~\ref{sec:prelim} introduces the needed background one
concept at a time. Section~\ref{sec:hspsep} separates unit-inversion hardness from hidden-subgroup
hardness, and Section~\ref{sec:notautomatic} explains why the result is not automatic from
Artin--Wedderburn. Section~\ref{sec:security} fixes the security model and
Section~\ref{sec:security-impl} draws the practical security implications. The technical core
follows. Sections~\ref{sec:main}--\ref{sec:dichotomy} build the algorithm, the
boundary test, and the complexity picture. Section~\ref{sec:coherent} gives the explicit circuit.
Sections~\ref{sec:e2e} and~\ref{sec:impact} report an end-to-end attack and its practical impact.
Sections~\ref{sec:surviving}--\ref{sec:resources} give the candidate construction, the experiments,
and the resource estimates. Sections~\ref{sec:related}--\ref{sec:conclusion} cover related work,
limitations, and conclusions.

\begin{figure}[t]\centering
\includegraphics[width=\linewidth]{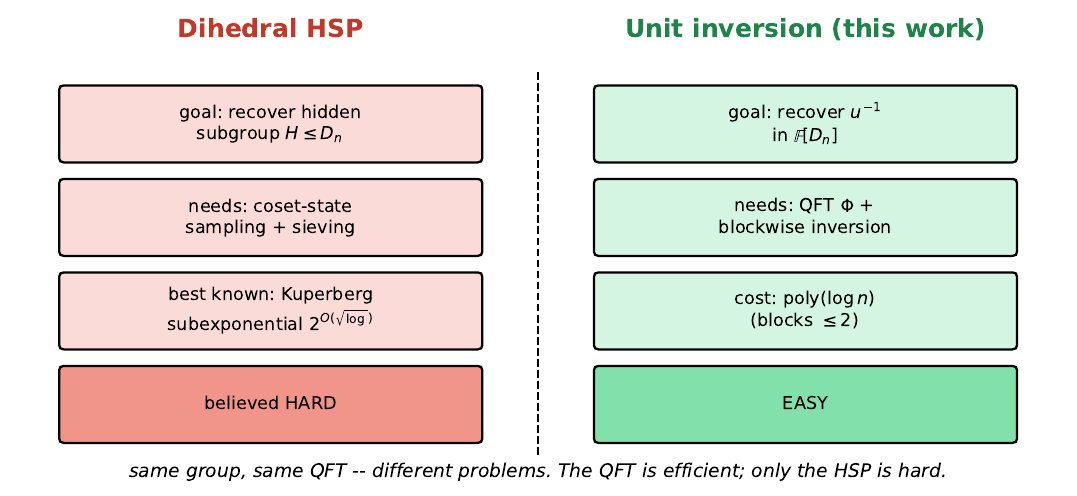}
\caption{The two problems on the dihedral group $D_n$. Both use the same efficient dihedral Fourier
transform. The hidden subgroup problem also requires recovering a hidden subgroup from coset states,
which is believed hard. Unit inversion does not, and is polynomial under the conditions of
Theorem~\ref{thm:gui}. Previous work often treats the two problems as one; this paper separates
them.}
\label{fig:hsp}
\end{figure}

\section{Background}\label{sec:prelim}

This section introduces the needed ideas one at a time, in the order the paper uses them. A reader
familiar with group rings and the Wedderburn decomposition may skip to Section~\ref{sec:security}.

\paragraph{Group rings and units} Let $R$ be a ring and let $G=\{g_0=e,g_1,\dots,g_{n-1}\}$ be a
finite group. The group ring $\RG$ consists of formal sums $\sum_g a_g g$ with coefficients
$a_g\in R$. Addition is componentwise and multiplication follows the group operation (convolution).
A \emph{unit} is an invertible element, and the units form the group $\RG^\times$. The schemes we
study~\cite{hurley2011,mittalkumar2021,idbased2022} assume that, given a unit $u$ presented as a
product, recovering $u^{-1}$ is hard.

\paragraph{Splitting a group ring into blocks} Let $\F$ be a field. By Maschke's theorem, the group
algebra $\FG$ is \emph{semisimple} exactly when the characteristic of $\F$ does not divide $|G|$.
When $\F$ is also a splitting field for $G$, the Artin--Wedderburn theorem gives an algebra
isomorphism
\begin{equation}\label{eq:aw}
\Phi:\FG\;\xrightarrow{\ \sim\ }\;\bigoplus_{t=1}^{h}M_{d_t}(\F),\qquad
\Phi(a)=\Big(\textstyle\sum_g a_g\,\rho_t(g)\Big)_{t=1}^{h},
\end{equation}
where $\rho_1,\dots,\rho_h$ are the inequivalent irreducible representations of dimensions
$d_1,\dots,d_h$, and $\sum_t d_t^2=|G|$. In words, the group ring becomes a list of small matrix
algebras. Any way of computing $\Phi$ is a \emph{generalized Fourier transform}. For abelian $G$,
all blocks have size $d_t=1$ and $\Phi$ is the ordinary character transform of~\cite{dooms2025}. A
quantum circuit for $\Phi$ is a QFT over $G$, and efficient circuits (polynomial in $\log|G|$) are
known for abelian groups, supersolvable groups~\cite{beals1997}, broad classes via the generic
construction~\cite{moore2006}, and explicitly for dihedral groups~\cite{dihedralgauge2024}.

\paragraph{The hidden subgroup problem, and what we do not use} The HSP asks one to recover an
unknown subgroup $H\le G$ from a function that is constant on the cosets of $H$. The abelian case is
easy. The non-abelian case is hard in general; the dihedral case admits only Kuperberg's
subexponential algorithm~\cite{kuperberg2005}, a fact linked to the shortest-vector
problem~\cite{regev2004}. Our algorithm solves no HSP. It uses the Fourier transform only as the
change of basis in Eq.~\eqref{eq:aw}. This difference is the separation we exploit, and it is the
reason a hard HSP does not protect these schemes.

\section{Separating hidden-subgroup hardness from unit-inversion hardness}
\label{sec:hspsep}

The non-abelian turn in group-ring cryptography rests on a single inference: that because the
dihedral hidden subgroup problem (HSP) resists efficient quantum algorithms, a cryptosystem built on
a dihedral group ring inherits that resistance. This section isolates that inference and shows it
does not hold.

\paragraph{What was implicitly assumed} The security of a non-abelian group-ring scheme was argued by
analogy. The abelian schemes fell because the abelian Fourier transform diagonalizes them, so a
platform whose associated hidden subgroup problem is hard should resist the analogous attack. The
implicit premise is that the attack on the abelian scheme \emph{was} an HSP algorithm, so that HSP
hardness would block its non-abelian analogue.

\paragraph{Why the assumption looks plausible} The premise is attractive for three reasons. The
abelian and non-abelian attacks share a subroutine, the Fourier transform over $G$, so they appear to
be instances of one technique. The dihedral HSP is one of the few natural quantum problems without a
polynomial algorithm, so it is a tempting security anchor. And the migration
abelian~$\to$~dihedral visibly increases algebraic structure, which is loosely associated with
hardness.

\paragraph{Why it is false} Unit inversion and the HSP share the transform and nothing else. The HSP
must, after transforming, recover a hidden subgroup from coset states; this is the step that resists
efficient algorithms, and for the dihedral group it is the entire
difficulty~\cite{kuperberg2005,regev2004}. Unit inversion forms no coset state and recovers no
subgroup. By Theorem~\ref{thm:gui} it applies $\Phi$, inverts each block, and applies $\Phi^{-1}$;
the hard step of the HSP is absent from the computation. The dihedral group is the standard witness
that an efficient Fourier transform need not yield an efficient HSP algorithm~\cite{kuperberg2005},
and unit inversion needs only the efficient half of that witness. Consequently dihedral-HSP hardness,
however solid, cannot protect a scheme whose security reduces to unit inversion: the two hardness
questions concern different computational problems.

\begin{proposition}[HSP hardness does not imply inversion hardness]\label{prop:independence}
There is a family of groups---the dihedral groups $D_n$---for which unit inversion is
quantum-polynomial (Theorem~\ref{thm:easy} and Corollary~\ref{cor:dihedral}) while the hidden
subgroup problem admits only subexponential-time algorithms~\cite{kuperberg2005}. Hence the hardness
of the HSP does not transfer to unit inversion.
\end{proposition}

\paragraph{Why this changes how group-ring schemes must be evaluated} The practical upshot is that an
HSP-based security argument is not evidence for a group-ring inversion scheme. A platform must instead
be judged by the two parameters that govern inversion: whether its group algebra is semisimple with
an efficient generalized transform, and whether its largest irreducible dimension is polynomially
bounded. Section~\ref{sec:impact} turns this into a checklist; the point here is that the checklist
does not contain the hidden subgroup problem.

\section{Why the cryptanalytic consequence is not automatic}\label{sec:notautomatic}

A natural first reaction is that our attack is a restatement of the Artin--Wedderburn theorem. We
address this directly, because the distinction is the point of the paper.

\paragraph{What Artin--Wedderburn provides} For a semisimple group algebra over a splitting field, the
theorem asserts the \emph{existence} of an algebra isomorphism
$\Phi:\FG\xrightarrow{\sim}\bigoplus_t M_{d_t}(\F)$. Inverting an element by inverting its blocks is
then immediate, and classical algorithms that compute such decompositions are well
established~\cite{friedlronyai,bremner2011}. At the level of classical existence, the inversion
identity (Theorem~\ref{thm:gui}(i)--(ii)) is elementary, and we claim no novelty for it.

\paragraph{What it does not provide} Artin--Wedderburn is silent on three questions that decide
whether a cryptosystem is broken (Table~\ref{tab:awgap}). First, it does not say which computational
problem a scheme's security actually rests on; in particular it does not say whether the hidden
subgroup problem, the stated source of these schemes' quantum resistance, is engaged by inversion at
all. Second, it is an existence statement, not a cost statement: it does not bound the size or depth
of a circuit realizing $\Phi$, nor identify the parameter that separates tractable from intractable
platforms. Third, it says nothing about realizing the isomorphism \emph{coherently} on a
fault-tolerant machine, with bounded width and clean ancillae, as a sub-routine of a larger quantum
computation.

\begin{table}[t]\centering
\caption{What Artin--Wedderburn supplies, and what the cryptanalytic consequence additionally
requires. The right column is the content of this paper.}
\label{tab:awgap}
\small
\renewcommand{\arraystretch}{1.3}
\begin{tabular}{p{3.7cm}p{4.3cm}}
\toprule
Artin--Wedderburn provides & The cryptanalysis additionally requires \\
\midrule
Existence of $\Phi:\FG\cong\bigoplus_t M_{d_t}(\F)$ & The observation that inversion factors through $\Phi$ alone and never invokes the HSP (\textsection\ref{sec:hspsep}) \\
Blockwise inversion of an element & The parameter $d_{\max}$ that separates easy from hard platforms (\textsection\ref{sec:dichotomy}) \\
A statement about algebras & A statement about a quantum trapdoor's security, with an exact boundary witness (Thm.~\ref{thm:boundary}) \\
No circuit & An explicit reversible realization with $T$-count and width (Thm.~\ref{thm:gui}(iii), \textsection\ref{sec:coherent}) \\
\bottomrule
\end{tabular}
\end{table}

\paragraph{Why the consequence is not automatic} The cryptanalytic content of this paper is precisely
the part Artin--Wedderburn omits. Three things must be added, and none follows from the existence of
$\Phi$:

\begin{enumerate}[label=(\roman*),leftmargin=2.0em,itemsep=1pt]
\item unit inversion factors through the transform alone and never invokes the hidden subgroup
problem (Section~\ref{sec:hspsep}), so the hardness the platforms were chosen for is irrelevant to
the task they pose;
\item tractability is governed by a single representation-theoretic parameter $d_{\max}$, giving a
complexity dichotomy (Section~\ref{sec:dichotomy}) rather than a mere algorithm;
\item the isomorphism admits an explicit, validated reversible realization with stated fault-tolerant
cost (Section~\ref{sec:coherent}).
\end{enumerate}

\paragraph{Why the literature did not draw this conclusion} The decomposition has been available for a
century, yet the schemes it dissolves were proposed, and migrated to non-abelian platforms, after it
was known. The reason is that the decomposition was read as a structural fact about algebras, not as a
statement about the quantum-security of a trapdoor. The security analyses reasoned about the hidden
subgroup problem, which the decomposition does not mention, and the computational-algebra literature
that does use the decomposition was not asking whether a particular cryptographic hardness assumption
survives a quantum adversary. Our contribution is to connect these two literatures and to observe
that the connection removes the assumed hardness rather than relying on it.

\section{Security Model and Adversary}\label{sec:security}

To state which assumptions are affected, and against which adversaries, we fix a model in the
standard provable-security vocabulary, written for group-ring primitives. The targeted
schemes~\cite{hurley2011,mittalkumar2021,idbased2022} share a common shape. A unit $u\in\FG^\times$
is a one-way trapdoor; public operations use maps built from $u$, while the secret functionality
needs $u^{-1}$. We isolate the underlying computational problem.

\begin{definition}[Unit Inversion problem, $\mathrm{UINV}_{G,\F}$]\label{def:uinv}
Given a finite group $G$, a field $\F$, and the coordinate vector of a unit $u\in\FG^\times$, output
the coordinate vector of $u^{-1}$.
\end{definition}

\begin{definition}[$\mathrm{UINV}$ game and advantage]\label{def:game}
For an algorithm $A$, the game $\mathrm{Game}^{\mathrm{UINV}}_{G,\F}(A)$ runs as follows. The
challenger samples $u\stackrel{\$}{\leftarrow}\FG^\times$ and sends $u$ to $A$; $A$ returns
$v\in\FG$; and $A$ wins if $uv=vu=1$. We write
$\mathrm{Adv}^{\mathrm{UINV}}_{G,\F}(A)=\Pr[A\text{ wins}]$.
\end{definition}

We distinguish two adversaries. A classical (PPT) adversary is the one the proposals are designed
against; its best known inversion is dense linear algebra in $\Theta(|G|^\omega)$ field operations.
A quantum (BQP) adversary is the class against which post-quantum security is claimed, and the one
our algorithm inhabits. We also separate three access levels: key recovery (KR), chosen-plaintext
(CPA), and chosen-ciphertext (CCA). Our algorithm operates at the KR surface, which is the weakest
and therefore the most consequential: from the public unit $u$ alone it returns $u^{-1}$, and a KR
break trivially gives CPA and CCA breaks.

\begin{proposition}[Polynomial-time key recovery under the model assumptions]\label{prop:kr}
Let $G$ be a finite group and $\F$ a splitting field with $\mathrm{char}\,\F\nmid|G|$, admitting an
efficient quantum Fourier transform with $d_{\max}=\poly(\log|G|)$. Then there is a quantum
polynomial-time adversary $A^*$ with $\mathrm{Adv}^{\mathrm{UINV}}_{G,\F}(A^*)=1$.
\end{proposition}
\begin{proof}
$A^*$ runs the inversion algorithm of Theorem~\ref{thm:gui}: it applies $\Phi$, inverts each block,
and applies $\Phi^{-1}$, returning $v=u^{-1}$. By Theorem~\ref{thm:gui}(i) the output satisfies
$uv=vu=1$ for every unit, so the success probability is $1$. By Theorem~\ref{thm:gui}(iii) and
Theorem~\ref{thm:reduction} the running time is polynomial in the input size under the stated
conditions.
\end{proof}

\begin{corollary}[Security implications for semisimple group-ring platforms]\label{cor:impl}
Let $S$ be a group-ring scheme over a platform satisfying the hypotheses of
Proposition~\ref{prop:kr}, whose security (KR, CPA, or CCA) is established by a reduction showing
that breaking $S$ is at least as hard as solving $\mathrm{UINV}_{G,\F}$. Then $S$ is insecure
against quantum adversaries under those conditions: $A^*$ recovers the trapdoor from public data
alone and wins the corresponding game with probability $1$, in time polynomial in the input size.
\end{corollary}

These conditions hold for the abelian platforms, recovering the result of~\cite{dooms2025}; for the
dihedral platforms advocated for quantum resistance (Corollary~\ref{cor:dihedral}); and for all
small non-abelian platforms of bounded block dimension. The regimes that escape are the modular case
(Section~\ref{sec:boundary}) and platforms whose representation dimension grows super-polynomially
(Section~\ref{sec:dichotomy}). Because the argument targets the hardness assumption rather than a
particular padding or transform, oracle-based armouring does not help, as we make precise in
Section~\ref{sec:security-impl}. We also note that the schemes
of~\cite{hurley2011,mittalkumar2021} carry no formal security reduction, so for those the
consequence is against the intended hardness assumption rather than against a published theorem.

\section{Security Implications and Affected Schemes}\label{sec:security-impl}

Before the detailed algorithm and proofs, we state what the result means for deployed and proposed
constructions, since this is what most readers will want first. The detailed cost analysis follows
in Sections~\ref{sec:main}--\ref{sec:dichotomy}.

Table~\ref{tab:affected} classifies the schemes by the three platform properties that matter for the
method: semisimplicity, an efficient QFT, and bounded block dimension. The abelian schemes are
covered by~\cite{dooms2025}; the dihedral and small non-abelian schemes by the present work; the
large-block platforms remain open by Conjecture~\ref{conj:hard}; and the modular regime escapes by
Theorem~\ref{thm:boundary}, with a candidate construction proposed in Section~\ref{sec:surviving}.
The operational message is simple: moving from an abelian to a dihedral or other small platform does
not, under these conditions, restore quantum security, and harvest-now-decrypt-later exposure should
be assessed on that basis.

\begin{table}[t]\centering
\caption{Security taxonomy of the affected schemes. ``ss.''~means semisimple and ``bdd.\
irr.''~means bounded irreducible (block) dimension; ``vulnerable'' is meant under the model
assumptions of Section~\ref{sec:security}. Verdicts follow Theorems~\ref{thm:gui},
\ref{thm:algebra}, and~\ref{thm:boundary} and Conjecture~\ref{conj:hard}.}
\label{tab:affected}
\footnotesize\setlength{\tabcolsep}{3pt}
\begin{tabular}{p{2.46cm}p{1.80cm}cccp{2.62cm}}
\toprule
Scheme family & Platform & ss. & QFT & bdd. irr. & Verdict (under model assumptions) \\
\midrule
Hurley--Hurley (2011) & abelian/cyclic & yes & yes & yes & vulnerable (parent + this work) \\
Mittal--Kumar (2021) & abelian/cyclic & yes & yes & yes & vulnerable (parent + this work) \\
matrix/ElGamal over group ring & dihedral $D_n$ & yes & yes & yes & vulnerable (this work) \\
ID-based group-ring encryption & small non-abelian & yes & yes & yes & vulnerable (this work) \\
hypothetical large-irrep platform & symmetric $S_m$ & yes & unknown & no & open (block step exp.) \\
hypothetical modular platform & any, $p\mid|G|$ & no & n/a & n/a & method prerequisite fails \\
\bottomrule
\end{tabular}
\end{table}

A natural defense to consider is a random-oracle transform, the usual route to chosen-ciphertext
security. We show it does not work here, and the reason is instructive: it is structural, not
quantitative. Such transforms are analyzed in the quantum random-oracle model (QROM), where the hash
function is modelled as an oracle a quantum adversary may query in superposition. A transform of this
kind hardens the wrapper around a trapdoor, but our attack defeats the trapdoor itself, so the
wrapper is bypassed entirely.

\begin{proposition}[QROM-invariance]\label{prop:qrom}
Let $S$ be a group-ring scheme over a platform satisfying the hypotheses of
Proposition~\ref{prop:kr}, and let $S_H=\mathsf{T}^H(S)$ be obtained by a random-oracle transform
$\mathsf{T}$ with quantum-accessible oracle $H$ (for instance the $\mathrm{FO}^{\bot}$
key-encapsulation-mechanism (KEM) compiler). Then there is an efficient adversary that makes $0$
queries to $H$ and recovers the decapsulation secret of $S_H$, winning key recovery with probability
$1$. Hence the QROM query complexity to break $S_H$ is $q_H=0$ for key recovery and $q_H\le 1$ for
session-key recovery on a specific challenge ciphertext, independent of any measure-and-reprogram
analysis.
\end{proposition}
\begin{proof}
The transform $\mathsf{T}$ wraps encapsulation and decapsulation with calls to $H$ but leaves the
public trapdoor unit $u$ in the public parameters. Running $A^*$ of Proposition~\ref{prop:kr} on $u$
outputs $u^{-1}$ in polynomial time using only group-ring arithmetic and $\Phi$, issuing no oracle
query, so key recovery succeeds with $q_H=0$. For IND-CCA on a fixed challenge, the attacker holding
$u^{-1}$ decapsulates $m$ and evaluates $K=H(m)$ with a single query. No property of $H$ is used.
\end{proof}

The reading is that random-oracle transforms protect the wrapper. They prevent the exploitation of
malleability or decryption-oracle leakage, while assuming a sound trapdoor. Since our method defeats
the trapdoor itself by algebraic structure, the wrapper is beside the point. The only defenses that
change the picture alter the algebraic setting so that $\mathrm{UINV}$ is no longer easy, by moving
to the non-semisimple regime or to super-polynomial representation dimension. Both are statements
about the platform, not about the oracle. A self-contained QROM background and the
$\mathrm{FO}^{\bot}$ specifics are given in Supplementary~\textsection S10.

\section{Efficient Inversion via the Generalized Fourier Transform}\label{sec:main}

The easy direction turns on identifying \emph{which} algebraic fact controls the cost---the block
decomposition---and on the observation that the quantum hardness anchor, the HSP, plays no role in
it. We state the fact first and then turn it into a cost estimate.

\begin{theorem}[Generalized unit inversion]\label{thm:gui}
Let $G$ be a finite group and $\F$ a splitting field for $G$ with $\mathrm{char}\,\F\nmid|G|$, and
let $\Phi$ be as in Eq.~\eqref{eq:aw}. Then:
\begin{enumerate}[label=(\roman*),leftmargin=2.2em]
\item $u\in\FG$ is a unit iff every block $\Phi_t(u)\in M_{d_t}(\F)$ is invertible, in which case
$\Phi(u^{-1})=(\Phi_t(u)^{-1})_{t=1}^{h}$.
\item Unit inversion therefore reduces to one evaluation of $\Phi$, $h$ matrix inversions of total
cost $\sum_t O(d_t^\omega)$, and one evaluation of $\Phi^{-1}$.
\item This reduction is realized by an explicit quantum circuit: the generalized QFT $\Phi$ followed
by the Coherent Block Inversion circuit of Section~\ref{sec:coherent}. When the QFT over $G$ has an
efficient (polylogarithmic-depth) circuit and $d_{\max}=\poly(\log|G|)$, the result has total
$T$-count $O(|G|\,d_{\max}\,\mathrm{polylog}\,q)$ and width $O(|G|\log q)$. Unit inversion is then
efficient, that is, polynomial in the input size $|G|\lceil\log_2 q\rceil$.
\end{enumerate}
\end{theorem}
\begin{proof}[Proof sketch; the full argument is given in Supplementary~\textsection S1.]
Over the splitting field, Maschke's theorem and Artin--Wedderburn make $\Phi$ an algebra
isomorphism, so it sends units to units with $\Phi(u^{-1})=\Phi(u)^{-1}$. Invertibility in a finite
product of matrix algebras is blockwise, which gives~(i). Part~(ii) is the algorithmic restatement:
evaluate $\Phi$, invert each block by Gaussian elimination, and evaluate $\Phi^{-1}$. For~(iii) the
QFT realizes $\Phi$ as a reversible $\F_q$-linear map, and each block is inverted by the reversible
circuit of Section~\ref{sec:coherent} on a disjoint sub-register, so the blocks invert in parallel.
This yields the stated $T$-count and width. The claim is efficiency in the input size, not
polylogarithmic total time, since the input already has size $\Theta(|G|\log q)$.
\end{proof}

\begin{remark}\label{rem:scope}
Parts~(i)--(ii) are a statement of classical algebra, proved without conditions and used as the
correctness oracle in our experiments. Part~(iii) is realized concretely by the circuit of
Section~\ref{sec:coherent}, with explicit gate, $T$-count, and qubit accounting.
\end{remark}

The dihedral case is the one most often advanced for quantum resistance. It is the sharpest test of
the separation in Section~\ref{sec:hspsep}, because here the security rationale (dihedral-HSP
hardness) and the actual task (unit inversion) come apart completely.

\begin{corollary}[Efficient inversion of units in dihedral group rings]\label{cor:dihedral}
Every irreducible representation of $D_n$ has dimension at most $2$, and the QFT over $D_n$ has a
circuit of size $\poly(\log n)$~\cite{dihedralgauge2024}. Hence unit inversion in $\F[D_n]$ is
efficient, polynomial in the input size, classically and as a fault-tolerant circuit, independently
of the hardness of the dihedral hidden subgroup problem. Group-ring cryptosystems instantiated over
dihedral platforms therefore do not attain quantum security under the conditions of
Theorem~\ref{thm:gui}.
\end{corollary}

We verified Theorem~\ref{thm:gui}(i)--(ii) and Corollary~\ref{cor:dihedral} exactly over $\F_p$. For
each tested $D_n$ we computed $u^{-1}$ by Wedderburn block inversion and confirmed both that
$u\,u^{-1}=e$ and that the result agrees with the inverse obtained from the full regular
representation. Table~\ref{tab:dihedral} reports complete agreement together with the operation
counts. The block route grows linearly in $|G|$ because all blocks have dimension at most two, while
the generic route is cubic. Two fully worked numerical examples, $\F_5[D_4]$ and $\F_7[D_6]$, with
the explicit block matrices before and after the transform, are given in
Supplementary~\textsection S3.

\begin{table}[t]\centering
\caption{Exact inversion of units in dihedral group rings over $\F_p$ ($p\equiv 1\bmod n$).
Wedderburn block inversion matches the generic inverse on every tested unit and is markedly cheaper.
``Wedderburn ops''~$=\sum_t d_t^3$ and ``regular ops''~$=|G|^3$.}
\label{tab:dihedral}
\begin{tabular}{ccccccc}
\toprule
$D_n$ & $\mathbb{F}_p$ & order & verified & $d_{\max}$ & Wedderburn ops & regular ops \\
\midrule
$D_{3}$ & $\mathbb{F}_{7}$ & 6 & 119/119 & 2 & 10 & 216 \\
$D_{5}$ & $\mathbb{F}_{11}$ & 10 & 142/142 & 2 & 18 & 1000 \\
$D_{7}$ & $\mathbb{F}_{29}$ & 14 & 175/175 & 2 & 26 & 2744 \\
$D_{11}$ & $\mathbb{F}_{23}$ & 22 & 138/138 & 2 & 42 & 10648 \\
$D_{13}$ & $\mathbb{F}_{53}$ & 26 & 164/164 & 2 & 50 & 17576 \\
\bottomrule
\end{tabular}
\end{table}

We now restate the reduction quantitatively, separating the unconditional classical bound from the
quantum claim that depends on an efficient transform.

\begin{theorem}[Reduction with explicit cost]\label{thm:reduction}
Let $\F$ be a splitting field for $G$ with $\mathrm{char}\,\F\nmid|G|$. Then $\mathrm{UINV}_{G,\F}$
reduces to one application of $\Phi$, one of $\Phi^{-1}$, and $h$ blockwise matrix inversions, with
total cost
\begin{equation}\label{eq:cost}
\begin{aligned}
T_{\mathrm{UINV}}(G,\F)&\;\le\;T_\Phi(G)+T_{\Phi^{-1}}(G)+\sum_{t=1}^h O(d_t^\omega)\\
&\;\le\;T_\Phi(G)+T_{\Phi^{-1}}(G)+O\!\big(d_{\max}^{\,\omega-2}|G|\big),
\end{aligned}
\end{equation}
where $\omega<2.371552$~\cite{williams2024}. In particular: (i) classically, using the fast
generalized DFT, $\mathrm{UINV}$ is solvable in
$\widetilde{O}(|G|\,d_{\max}+d_{\max}^{\omega-2}|G|)$ field operations, near-linear in $|G|$ for
bounded $d_{\max}$, to be compared with $\Theta(|G|^\omega)$ for dense inversion; (ii) quantumly, if
$\Phi$ admits a circuit of size $\poly(\log|G|)$ and $d_{\max}=\poly(\log|G|)$, then
$\mathrm{UINV}_{G,\F}$ is solved by a quantum circuit whose transform stage has depth
$\poly(\log|G|)$ and whose block stage uses $\poly(d_{\max},\log|\F|)$ reversible arithmetic per
block.
\end{theorem}
\begin{proof}[Proof; the full version is given in Supplementary~\textsection S2.]
Correctness is Theorem~\ref{thm:gui}(i). The second inequality uses that $x\mapsto x^{\omega-2}$ is
nondecreasing together with the Wedderburn identity:
$\sum_t d_t^\omega=\sum_t d_t^{\omega-2}d_t^2\le d_{\max}^{\omega-2}\sum_t d_t^2=d_{\max}^{\omega-2}|G|$.
Claim~(i) substitutes the fast generalized DFT bound for $T_\Phi$; claim~(ii) substitutes a
$\poly(\log|G|)$ quantum circuit and the coherent block stage of Section~\ref{sec:coherent}.
\end{proof}

\section{Generalization to Semisimple Algebras}\label{sec:algebra}

The argument of Section~\ref{sec:main} uses nothing specific to group rings beyond the
decomposition in Eq.~\eqref{eq:aw}. The same principle therefore governs any finite-dimensional
semisimple algebra, which is why several previously unrelated breaks turn out to be instances of one
mechanism rather than independent results.

\begin{theorem}[Algebra inversion]\label{thm:algebra}
Let $A$ be a finite-dimensional semisimple algebra over a splitting field $\F$, so
$A\cong\bigoplus_t M_{d_t}(\F)$, and let $\Psi$ realize this isomorphism. Then $u\in A$ is a unit iff
each block $\Psi_t(u)$ is invertible, $\Psi(u^{-1})=(\Psi_t(u)^{-1})_t$, and unit inversion costs one
evaluation of $\Psi^{\pm1}$ plus $\sum_t O(d_t^\omega)$. If $\Psi$ has an efficient quantum circuit
and $d_{\max}=\poly(\log\dim A)$, unit inversion in $A$ is quantum polynomial time.
\end{theorem}

Theorem~\ref{thm:algebra} covers twisted group algebras, crossed products, and Hecke-type algebras,
and in particular the finite-dimensional associative-algebra signature schemes
of~\cite{romankov2023}. These become instances of one principle: a cryptosystem whose hard task is
inversion in a semisimple algebra with an efficient generalized Fourier transform does not attain
quantum security. The classical machinery for computing Wedderburn decompositions is well
established~\cite{bremner2011,friedlronyai}. What was not established is that this classical
machinery, realized coherently, dissolves the quantum-security claim of an entire platform family;
our contribution is to draw that consequence and to give an efficient realization for structured
platforms.

\section{The Semisimplicity Boundary}\label{sec:boundary}

The decomposition in Eq.~\eqref{eq:aw}, and hence the whole method, requires semisimplicity. The
point of this section is that the failure of semisimplicity is not just a qualitative caveat. It is
an exact, efficiently checkable condition.

\begin{theorem}[Modular boundary]\label{thm:boundary}
Let $\F$ have characteristic $p$ with $p\mid|G|$. Then $\FG$ is not semisimple: its Jacobson radical
is nonzero, no isomorphism $\FG\cong\bigoplus_t M_{d_t}(\F)$ exists, and the generalized Fourier
transform $\Phi$ is unavailable. Moreover, writing $L$ for the left-regular representation and
$T_{ij}=\tr_L(e_ie_j)$ for the regular trace form on the group basis,
\begin{equation}\label{eq:traceform}
T=|G|\,P,\qquad \det T=\pm|G|^{|G|},
\end{equation}
where $P$ is the permutation matrix of $g\mapsto g^{-1}$. Hence over $\F_p$, $\det T\equiv 0\pmod p$
if and only if $p\mid|G|$.
\end{theorem}
\begin{proof}
Maschke's theorem is an equivalence, so when $p\mid|G|$ the radical is nonzero and $\FG$ admits no
decomposition into matrix algebras over $\F$, whence $\Phi$ does not exist. For the trace form,
$\tr L(g)=\#\{h:gh=h\}=|G|\,[g=e]$, so $T_{ij}=|G|\,[g_j=g_i^{-1}]$, i.e.\ $T=|G|P$ with $P$ an
involution and $\det P=\pm1$. Therefore $\det T=|G|^{|G|}\det P=\pm|G|^{|G|}$, and reducing modulo
$p$ gives $\det T\equiv0$ exactly when $p\mid|G|$.
\end{proof}

Equation~\eqref{eq:traceform} gives an exact witness of the boundary, computable directly from the
group basis. Figure~\ref{fig:boundary} shows the resulting semisimplicity grid for cyclic groups
across characteristics. The trace-form determinant vanishes precisely on the cells with $p\mid|G|$,
in agreement with the theorem on every cell. The same computation identifies the only regimes in
which a group-ring scheme can avoid the method: the modular regime $p\mid|G|$ and the large-block
regime of Section~\ref{sec:dichotomy}. The extension of the boundary to non-cyclic dihedral
platforms is reported in Supplementary~\textsection S8. We stress that the failure of the method on a
given platform is not a proof of security; it removes the present obstruction but says nothing about
others.

\begin{figure}[t]\centering
\includegraphics[width=0.52\linewidth]{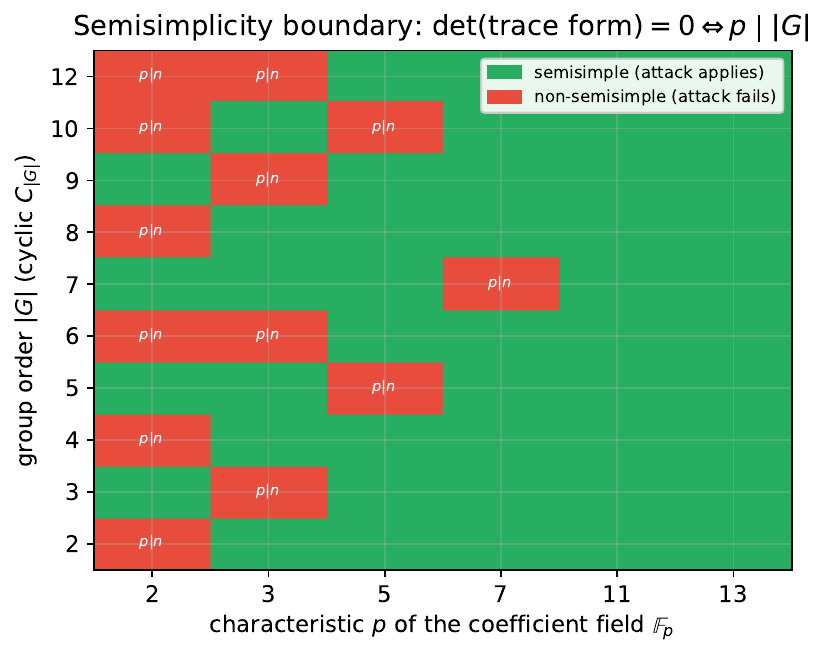}
\caption{Computed semisimplicity boundary for $\F_p[C_{|G|}]$. The trace-form determinant vanishes
precisely on the cells with $p\mid|G|$ (Theorem~\ref{thm:boundary}), where the Wedderburn
decomposition, and hence the method, does not exist.}
\label{fig:boundary}
\end{figure}

\section{A Complexity Dichotomy}\label{sec:dichotomy}

The results so far point to a single organizing principle: the cryptanalytic tractability of
group-ring unit inversion is governed by the largest irreducible dimension $d_{\max}$, together with
the efficiency of the transform, and \emph{not} by whether the group is abelian. Commutativity is the
wrong axis; block dimension is the right one. This section makes the principle precise as a
dichotomy, illustrated in Figure~\ref{fig:axis}.

\begin{figure}[t]\centering
\includegraphics[width=0.92\linewidth]{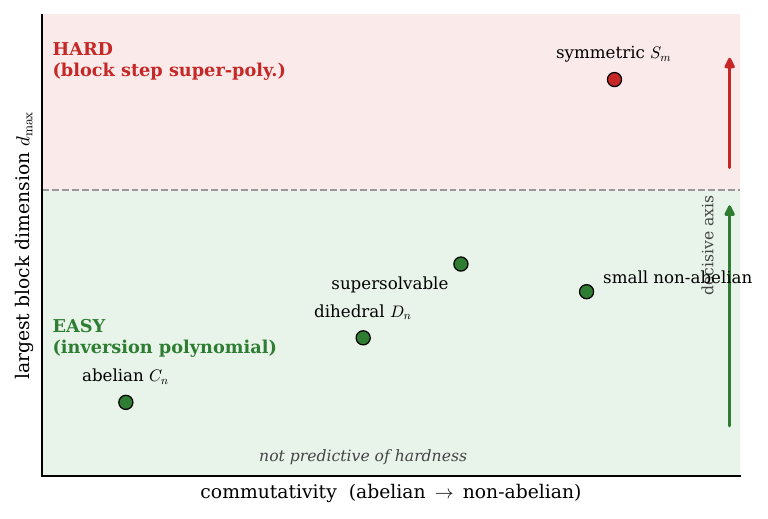}
\caption{Commutativity is not predictive of hardness; block dimension is. Every platform with
bounded $d_{\max}$ falls in the easy (green) band, regardless of where it sits on the
abelian--non-abelian axis; only super-polynomial $d_{\max}$ (the symmetric groups) reaches the hard
(red) band. The migration abelian~$\to$~dihedral moves horizontally and does not cross the
boundary that matters.}
\label{fig:axis}
\end{figure}

Taken together, the preceding sections mark out where unit inversion is easy. We proved the easy
direction above. Here we treat the hard direction, deferring the proof machinery (the output-size
separation and the \textsc{\#P}-hardness argument) to Supplementary~\textsection S5--S7 and keeping
the main statements in view.

\begin{theorem}[Easy direction]\label{thm:easy}
If $\FG$ (or a finite-dimensional algebra $A$) is semisimple, admits an efficient generalized
Fourier transform, and has $d_{\max}=\poly(\log|G|)$, then unit inversion is efficient: it is solved
in time polynomial in the input size by the reduction of Theorem~\ref{thm:reduction}, realized as
the circuit of Section~\ref{sec:coherent}.
\end{theorem}

The hard side has two independent sources, both visible in the cost model of
Figure~\ref{fig:dichotomy}: the largest block dimension $d_{\max}$ and the availability of an
efficient transform. The symmetric groups $S_m$ illustrate the first. There $d_{\max}$ equals the
largest number of standard Young tableaux of a partition of $m$ and grows super-polynomially
(Table~\ref{tab:dmax}), so the block-inversion step alone is super-polynomial; and for general groups
no efficient QFT is known. We give the hard direction concrete content through three results of
deliberately different logical strength.

The first is unconditional and concerns output size. A short unit can have a dense inverse, so no
algorithm can even write the answer quickly.

\begin{proposition}[Output-size separation]\label{prop:outputsize}
There are units $u\in\F_q[S_m]$ with $O(1)$ nonzero coordinates whose inverse $u^{-1}$ has support a
constant fraction of $|S_m|=m!$. Consequently, in the succinct-input model no algorithm, classical
or quantum, outputs $u^{-1}$ in the coordinate basis in time $\poly(\log|G|)$, because the output has
super-polynomial size. (Verified for $m=3,4,5$, where the inverse support fills the entire group;
see Supplementary~\textsection S5.)
\end{proposition}

The second is unconditional within a natural oracle model that grants the algorithm the transformed
blocks for free, isolating the cost of inverting the largest block.

\begin{theorem}[Block-oracle quantum query lower bound]\label{thm:blockoracle}
In the block-oracle model, any quantum algorithm that outputs the inverse unit must invert the
largest block $\Phi_{t^*}(u)\in M_{d_{\max}}(\F_q)$, which requires $\Omega(d_{\max})$ quantum
queries (and $\Omega(d_{\max}^2)$ to write the inverse block). For $G=S_m$, where $d_{\max}$ grows
super-polynomially in $\log|G|$ (Table~\ref{tab:dmax}), the block stage is therefore
super-polynomial.
\end{theorem}
\begin{proof}[Proof idea; the full proof is given in Supplementary~\textsection S6.]
By Theorem~\ref{thm:gui}(i), $\Phi(u^{-1})_{t^*}=\Phi_{t^*}(u)^{-1}$. Each entry of a $d\times d$
inverse depends on all $d^2$ entries, and an $N'$-query quantum algorithm computes amplitudes that
are degree-$2N'$ polynomials in the oracle entries. The polynomial method then forces
$N'=\Omega(d)$ per output entry. Substituting $d=d_{\max}(S_m)$ gives the claim.
\end{proof}

The third is conditional and shows that the specific route we use, evaluating Wedderburn data, is
itself intractable for $S_m$ unless a standard complexity collapse occurs. The link is an exact
identity: the character-weighted coordinate that the Wedderburn map assigns to the irrep $\chi_t$ of
$S_m$, evaluated on the matrix of group-basis coefficients, is the immanant
$\mathrm{imm}_{\chi_t}(\cdot)=\sum_{\sigma\in S_m}\chi_t(\sigma)\prod_i a_{i,\sigma(i)}$. The
permanent ($\chi_t$ trivial) and the determinant ($\chi_t$ the sign) are the two extreme cases, and
the family interpolates between them as the partition varies. This is what makes the route hard for
large irreps, and it is a representation-theoretic obstruction, not an artifact of our algorithm.

\begin{theorem}[Conditional \textsc{\#P}-hardness of the Wedderburn route]\label{thm:sharpP}
The character-weighted coordinates that a Wedderburn-based inverter must evaluate for $\F[S_m]$
include the immanant $\mathrm{imm}_{\chi_t}(\cdot)$ as a special case. By B\"urgisser's
theorem~\cite{burgisser2000} the immanants for partitions of unbounded Durfee size are
\textsc{\#P}-hard (indeed VNP-complete). Hence, unless $\mathrm{FP}=\textsc{\#P}$, no polynomial-time
algorithm evaluates the Wedderburn data for $\F[S_m]$ across all irreps, and the
block-diagonalization route to inversion is \textsc{\#P}-hard in the large-block regime.
\end{theorem}

We emphasize what these three results do and do not establish. They show that both ingredients the
easy direction relies on, namely small blocks and an efficient transform, provably fail for $S_m$,
and they rule out the natural Wedderburn route. They do not rule out a fundamentally different
algorithm in the explicit-input model. Closing that gap is the open content of the following
conjecture, on which we never rely elsewhere.

\begin{conjecture}[Hard direction]\label{conj:hard}
For families of finite groups (or semisimple algebras) with $d_{\max}$ super-polynomial in
$\log|G|$, or for which no efficient generalized QFT exists, unit inversion admits no efficient
(quantum or classical) algorithm. Together with Theorem~\ref{thm:easy} this would give the
equivalence: unit inversion is efficient if and only if the algebra is semisimple with an efficiently
Fourier-transformable, polynomially bounded block structure.
\end{conjecture}

\begin{table}[t]\centering
\caption{Large-block evidence for the hard direction in $\F[S_m]$: the maximal irreducible dimension
$d_{\max}$ (block-oracle cost $\Omega(d_{\max})$) and its cube (block-inversion field operations).
All values are exact.}
\label{tab:dmax}
\begin{tabular}{cccc}\toprule
$m$ & $|S_m|=m!$ & $d_{\max}$ & block cost $\sim d_{\max}^3$ \\\midrule

3 & 6 & 2 & 8.00e+00 \\
4 & 24 & 3 & 2.70e+01 \\
5 & 120 & 6 & 2.16e+02 \\
6 & 720 & 16 & 4.10e+03 \\
7 & 5040 & 35 & 4.29e+04 \\
8 & 40320 & 90 & 7.29e+05 \\
9 & 362880 & 216 & 1.01e+07 \\
10 & 3628800 & 768 & 4.53e+08 \\
\bottomrule
\end{tabular}
\end{table}

\begin{figure}[t]\centering
\includegraphics[width=0.80\linewidth]{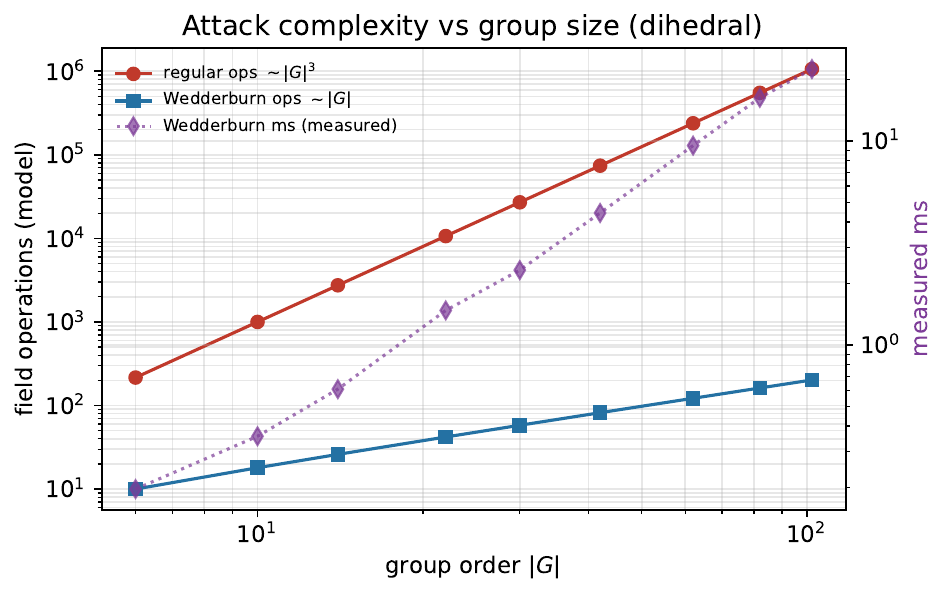}
\caption{The two faces of the dichotomy. (a) For dihedral groups all blocks have dimension at most
two, so Wedderburn inversion is near-linear while the generic method is cubic. (b) For symmetric
groups the maximal irreducible dimension grows super-polynomially. Tractability is governed by block
dimension, not by non-commutativity.}
\label{fig:dichotomy}
\end{figure}

\subsection{Proven, conditional, and open claims}\label{sec:proven}

Because the paper interleaves results of different logical status, we summarize that status
explicitly. The following are proved without conditions:

\begin{itemize}[leftmargin=1.4em,itemsep=1pt]
\item the algebraic inversion identity (Theorem~\ref{thm:gui}(i)--(ii));
\item the classical near-linear cost (Theorem~\ref{thm:reduction}(i));
\item the semisimple generalization (Theorem~\ref{thm:algebra});
\item the trace-form witness of the boundary (Theorem~\ref{thm:boundary});
\item the easy direction (Theorem~\ref{thm:easy});
\item the output-size separation (Proposition~\ref{prop:outputsize});
\item the correctness, unitarity, and clean-ancilla properties of the Coherent Block Inversion
circuit (Section~\ref{sec:coherent}, Supplementary~\textsection S1).
\end{itemize}

\noindent The block-oracle bound (Theorem~\ref{thm:blockoracle}) is
unconditional within its oracle model. Two claims are conditional: the \textsc{\#P}-hardness of the
Wedderburn route (Theorem~\ref{thm:sharpP}, under $\mathrm{FP}\ne\textsc{\#P}$) and the quantum
efficiency of Theorems~\ref{thm:gui}(iii) and~\ref{thm:reduction}(ii) (under an efficient
non-abelian QFT and bounded $d_{\max}$). All fault-tolerant figures in Section~\ref{sec:resources}
are model-based projections rather than measurements. Finally, the general unconditional lower bound
(Conjecture~\ref{conj:hard}), the hardness of the Radical Recovery assumption introduced in
Section~\ref{sec:surviving}, and on-hardware execution remain open. Every unconditional result in
the paper concerns the easy direction and the boundary.

For quick reference, Table~\ref{tab:status} restates this accounting in tabular form, separating
proven results from conditional ones and from open conjectures. Table~\ref{tab:assumptions} lists,
for the reader who wishes to audit the cryptanalytic conclusion, the four assumptions on which the
quantum break depends, where each is known to hold, and what follows if it fails. We regard these two
tables as the honest summary of how far the attack reaches: the cryptanalytic claim is firm exactly
where all four assumptions hold (the abelian, dihedral, and small non-abelian platforms), and is
correspondingly weaker elsewhere.

\begin{table}[t]\centering
\caption{Status of the paper's claims. ``Proven'' means established unconditionally (or
unconditionally within a stated oracle model); ``conditional'' means contingent on a named
assumption; ``open'' means conjectured or not yet established.}
\label{tab:status}
\footnotesize\setlength{\tabcolsep}{3pt}
\renewcommand{\arraystretch}{1.2}
\begin{tabular}{p{1.72cm}p{6.72cm}p{3.77cm}}
\toprule
Status & Claim & Location / qualifier \\
\midrule
Proven & Blockwise inversion identity $\Phi(u^{-1})=(\Phi_t(u)^{-1})_t$ over a splitting field & Thm.~\ref{thm:gui}(i)--(ii) \\
Proven & Classical near-linear cost $\widetilde{O}(d_{\max}^{\omega-2}|G|)$ & Thm.~\ref{thm:reduction}(i) \\
Proven & Generalization to semisimple algebras & Thm.~\ref{thm:algebra} \\
Proven & Exact trace-form witness of the semisimplicity boundary & Thm.~\ref{thm:boundary}, Eq.~\eqref{eq:traceform} \\
Proven & Easy direction of the dichotomy & Thm.~\ref{thm:easy} \\
Proven & Output-size separation for $S_m$ (succinct unit, dense inverse) & Prop.~\ref{prop:outputsize} \\
Proven & Correctness, unitarity, clean-ancilla behaviour of the block circuit & Lem.~\ref{lem:cbi}, Supp.~\textsection S1 \\
Proven (oracle model) & $\Omega(d_{\max})$ block-oracle query bound & Thm.~\ref{thm:blockoracle} \\
Conditional & Quantum efficiency of the attack & needs efficient non-abelian QFT and $d_{\max}=\poly(\log|G|)$; Thm.~\ref{thm:gui}(iii), \ref{thm:reduction}(ii) \\
Conditional & \textsc{\#P}-hardness of the Wedderburn route for $S_m$ & under $\mathrm{FP}\ne\textsc{\#P}$; Thm.~\ref{thm:sharpP} \\
Conditional & Fault-tolerant resource figures & model-based surface-code projections; \textsection\ref{sec:resources} \\
Open & General unconditional hardness of large-block / no-QFT platforms & Conj.~\ref{conj:hard} \\
Open & Hardness of the Radical Recovery assumption & Def.~\ref{def:rr}, Rem.~\ref{rem:rr} \\
Open & On-device (non-simulated) execution of the attack & \textsection\ref{sec:resources} \\
\bottomrule
\end{tabular}
\end{table}

\begin{table}[t]\centering
\caption{The four assumptions behind the quantum break, where each holds, and the consequence of its
failure. The cryptanalytic conclusion is firm only where all four hold at once.}
\label{tab:assumptions}
\footnotesize\setlength{\tabcolsep}{3pt}
\renewcommand{\arraystretch}{1.2}
\begin{tabular}{p{2.46cm}p{3.77cm}p{2.79cm}p{2.95cm}}
\toprule
Assumption & Holds for & Fails / unknown for & If it fails \\
\midrule
Efficient non-abelian QFT & abelian, dihedral, supersolvable groups \cite{beals1997,moore2006,dihedralgauge2024} & general $G$ (e.g.\ $S_m$) & no efficient quantum attack from our results \\
Splitting field, $\mathrm{char}\,\F\nmid|G|$ & $p\equiv1\bmod n$ for $\F_p[D_n]$ & non-splitting fields & blocks live over $\F_{p^k}$; base-field route still inverts (Supp.~\textsection S8) \\
Bounded block dimension $d_{\max}=\poly(\log|G|)$ & dihedral ($d_{\max}\le2$), small non-abelian & $S_m$ ($d_{\max}$ super-poly.) & block stage becomes super-polynomial (\textsection\ref{sec:dichotomy}) \\
Conjecture~\ref{conj:hard} (hard direction) & --- (conjectural) & --- & affects only claims about which platforms are \emph{secure}, not the attack \\
\bottomrule
\end{tabular}
\end{table}

\section{The Coherent Wedderburn Inversion Algorithm}\label{sec:coherent}

Theorem~\ref{thm:reduction} bounds the classical cost of inversion, but turning the reduction into a
genuine quantum algorithm requires one ingredient that is easy to overlook: the block inversions
must be performed coherently, without intermediate measurement and without leaving entangled ancilla
behind. We supply this as a self-contained reversible primitive, the Coherent Block Inversion
circuit, which is of independent interest wherever a finite-field matrix inverse must be applied
inside a larger quantum computation. The register layout, the five-step compute--copy--uncompute
construction, the proof of unitarity and correctness, the no-leakage argument, and the circuit
diagram are given in Supplementary~\textsection S1; here we record the resulting guarantees.

\begin{lemma}[Coherent Block Inversion]\label{lem:cbi}
Let $q$ be a prime power and $n_q=\lceil\log_2 q\rceil$. There is a circuit on $(3d^2+O(1))\,n_q$
qubits that preserves the input $M$, returns all scratch and ancilla to $|0\rangle$, and writes
$M^{-1}$ into a fresh output register whenever $M$ is invertible, flagging singular $M$ on a status
qubit. With the reversible modular-arithmetic circuits of~\cite{roetteler2017} its cost is
\[
\begin{aligned}
T(d,q)&=O\!\big(d^3 n_q^2+d\,n_q^3\big)\ T\text{-gates}=O(d^3\,\mathrm{polylog}\,q),\\
&\quad\text{width }(3d^2+O(1))\,n_q.
\end{aligned}
\]
\end{lemma}

\begin{lemma}[Leak-free composition]\label{lem:compose}
Applying the generalized QFT $\Phi$ to the coordinate register, then the circuit of
Lemma~\ref{lem:cbi} to each Wedderburn block on its own disjoint sub-register, and finally
$\Phi^{-1}$, yields a unitary mapping the coordinates of a unit $u$ to those of $u^{-1}$. Because
each block circuit returns its scratch to $|0\rangle$, no ancilla is shared or entangled across
blocks or with the transforms, and the blocks execute in parallel.
\end{lemma}

The composition of Lemma~\ref{lem:compose} is the Coherent Wedderburn Inversion algorithm, shown
end to end in Figure~\ref{fig:pipeline}. Its total
logical $T$-count is the two transforms plus $\sum_t T(d_t,q)=O(|G|\,d_{\max}\,\mathrm{polylog}\,q)$,
and its width is $O(|G|\log q)$. We validated the primitive with a register-level reversible
simulator that executes the construction on explicit $\F_q$ registers, logs every elementary
operation, and replays the inverse log to confirm scratch restoration. Across all tested dimensions
$d\le 6$ and fields $q\in\{29,257\}$, with $40$ random invertible matrices each, the simulator
confirmed without exception that the output equals $M^{-1}$, the input is preserved, and every
scratch and ancilla register returns exactly to $|0\rangle$. Representative resource counts and the
closed-form $T$-count are collected in Supplementary~\textsection S9.

\begin{figure}[t]\centering
\includegraphics[width=0.92\linewidth]{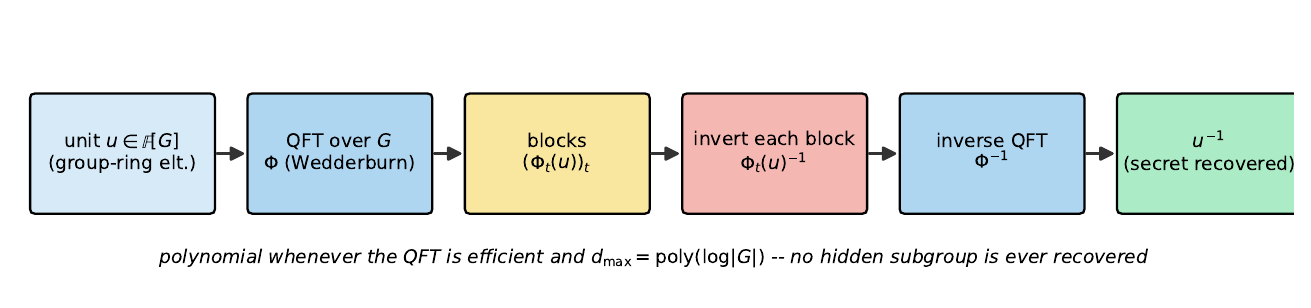}
\caption{The complete attack pipeline. A unit is mapped to block-diagonal form by the generalized
Fourier transform, each block is inverted independently by the Coherent Block Inversion circuit, and
the inverse transform returns the coordinates of $u^{-1}$. No coset state is formed and no hidden
subgroup is recovered at any stage.}
\label{fig:pipeline}
\end{figure}

\section{End-to-End Key Recovery}\label{sec:e2e}

To confirm that the $\mathrm{UINV}$ result translates into a concrete attack, we carried it through
to plaintext recovery against a representative group-ring public-key scheme in the
Hurley--Hurley / Mittal--Kumar trapdoor style. The private key is a unit $u\in\F_p[D_n]$, the public
key is the group-ring matrix $M=L(u)$, encryption is $c=u\cdot m$, and decryption is
$m=u^{-1}\cdot c$. Given only the public key, the attacker reads $u=Me_0$ from the first column,
applies the inversion of Section~\ref{sec:main} to recover $u^{-1}$, and decrypts.
Table~\ref{tab:keyrecovery} records the outcome: in every tested instance the recovered private key
matched the true key exactly, and all $500$ randomly generated ciphertexts per instance decrypted
correctly. This is a full public-key-to-plaintext recovery rather than a mere $\mathrm{UINV}$ oracle.

\begin{table}[t]\centering
\caption{End-to-end key recovery and decryption against the representative scheme. Since the schemes
of~\cite{hurley2011,mittalkumar2021} fix no standardized parameter sets and provide no formal
reduction, this instantiation is representative; the precise claim is that the method defeats the
intended hardness assumption.}
\label{tab:keyrecovery}
\begin{tabular}{ccccc}\toprule
Scheme instance & order & private key recovered & messages decrypted & rate \\\midrule

$\mathbb{F}_{7}[D_{3}]$ & 6 & yes & 500/500 & 1.000 \\
$\mathbb{F}_{11}[D_{5}]$ & 10 & yes & 500/500 & 1.000 \\
$\mathbb{F}_{29}[D_{7}]$ & 14 & yes & 500/500 & 1.000 \\
$\mathbb{F}_{23}[D_{11}]$ & 22 & yes & 500/500 & 1.000 \\
\bottomrule
\end{tabular}
\end{table}

Validation over genuinely non-splitting coefficient fields (where $p\not\equiv1\bmod n$ and the
two-dimensional blocks live over an extension $\F_{p^k}$) and over non-cyclic boundaries is reported
in Supplementary~\textsection S8. In all such cases the base-field regular-representation route,
which never invokes splitting, inverts every unit exactly, confirming that the splitting hypothesis
is a convenience for presenting the blocks rather than a precondition of the attack.

\section{Practical Cryptographic Impact}\label{sec:impact}

The preceding sections are algorithmic; this one asks what they mean for a practitioner deciding
whether a group-ring construction is safe to deploy against a future quantum adversary. We answer at
the level of the published proposals and their stated parameters, and we are careful to separate
``the trapdoor is inverted'' from ``the scheme is broken,'' since the two coincide only when the
scheme's security actually rests on unit inversion.

\paragraph{Comparison against published proposals and their parameters} Two lines of work are
usually cited as the basis for group-ring public-key cryptography. The first is the constructions of
Hurley and Hurley~\cite{hurley2011} (``group ring cryptography''), whose hard problem is precisely
the inverse computation problem, our $\mathrm{UINV}$. The second is the ElGamal-type group-ring
cryptosystems of Mittal et al.~\cite{mittalkumar2021} (``group rings based public key
cryptosystems''), together with the identity-based variant~\cite{idbased2022}. A concrete recent
instantiation with explicit parameter sets is the twisted dihedral-algebra key-encapsulation
mechanism of de la Cruz and Villanueva-Polanco~\cite{delacruz2021}, which proposes
$(p,n)\in\{(19,19),(23,23),(31,31),(41,41)\}$ targeting $128$/$192$/$256$-bit keys. None of the
first group fix standardized parameters; their papers give small illustrative examples and argue
security informally. Table~\ref{tab:realparams} therefore compares them by \emph{platform and
underlying hard problem}, with an explicit applicability column. The pattern is consistent. Where
security is defined by inverting a unit, our algorithm applies under the assumptions of
Table~\ref{tab:assumptions}. Where security instead rests on a discrete-logarithm or a decoding /
decisional assumption layered on top of the group ring, unit inversion is needed for legitimate
decryption but is not by itself the security claim, and our result is then a partial, not a total,
statement.

\begin{table}[t]\centering
\caption{Published group-ring proposals, their platforms and stated parameters, the problem their
security rests on, and the applicability of our result. ``ICP'' is the inverse computation problem
($\mathrm{UINV}$); ``DLPGR'' is the discrete logarithm in a group ring. Where no standardized
parameters are fixed, we note the platforms the source uses. Applicability is stated under the
assumptions of Table~\ref{tab:assumptions}; ``partial'' marks schemes whose stated security is not
unit inversion.}
\label{tab:realparams}
\footnotesize\setlength{\tabcolsep}{3pt}
\renewcommand{\arraystretch}{1.2}
\begin{tabular}{p{2.87cm}p{2.46cm}p{2.46cm}p{1.97cm}p{2.46cm}}
\toprule
Proposal & Platform / parameters & Security rests on & Semisimple, bdd.\ irr.? & Applicability of our result \\
\midrule
Hurley--Hurley, group ring cryptography~\cite{hurley2011} & $RG$, $R\in\{\Z,\Z_q,\mathrm{GF}(q)\}$; illustrative small $G$ & ICP ($=\mathrm{UINV}$), optionally $+$DLPGR & yes when $\mathrm{char}\nmid|G|$ and bdd.\ irr. & applies to the ICP layer; if DLPGR is added, that layer needs Shor-type analysis \\
Mittal et al., group rings based PKC~\cite{mittalkumar2021} & ElGamal / elliptic-ElGamal over $RG$; illustrative examples & DLP / ECDLP in $RG$ & yes & partial: legitimate decryption uses $u^{-1}$, but the security claim is DLP-type \\
ID-based group-ring encryption~\cite{idbased2022} & small non-abelian $G$ & unit-based trapdoor & yes & applies under the model assumptions \\
Twisted dihedral-algebra KEM~\cite{delacruz2021} & $\F_p[D_n]$ twisted, $(p,n)\!\in\!\{(19,19),\dots,(41,41)\}$, $128/192/256$-bit & decisional decoding in the twisted algebra & semisimple, $d_{\max}\!\le\!2$ & platform inversion is efficient (Cor.~\ref{cor:dihedral}); KEM security rests on a separate decisional assumption, \emph{not} directly broken \\
\bottomrule
\end{tabular}
\end{table}

\paragraph{What this does and does not imply for deployment} For schemes whose trapdoor is unit
inversion over a semisimple, bounded-dimension platform, the practical consequence is direct. The
private key is recoverable from public data by a quantum adversary under the stated assumptions. Such
schemes therefore provide no post-quantum confidentiality, and data encrypted under them today is
exposed to a harvest-now-decrypt-later strategy. The migration that the literature recommends, from
abelian to dihedral or other small non-abelian platforms, does not address this, because the
obstruction is the existence of an efficient block-diagonalizing transform, which dihedral groups
possess. By contrast, some proposals rest on a discrete-logarithm or a decisional/decoding assumption
(the ElGamal-type and the twisted-algebra KEM rows of Table~\ref{tab:realparams}). For those, our
result speaks only to the ease of inversion, not to the advertised hardness; breaking them would
require attacking the actual assumption, which we do not claim to do. Figure~\ref{fig:decision}
condenses this into a decision procedure a designer can apply to a candidate platform.

\begin{figure}[t]\centering
\includegraphics[width=\linewidth]{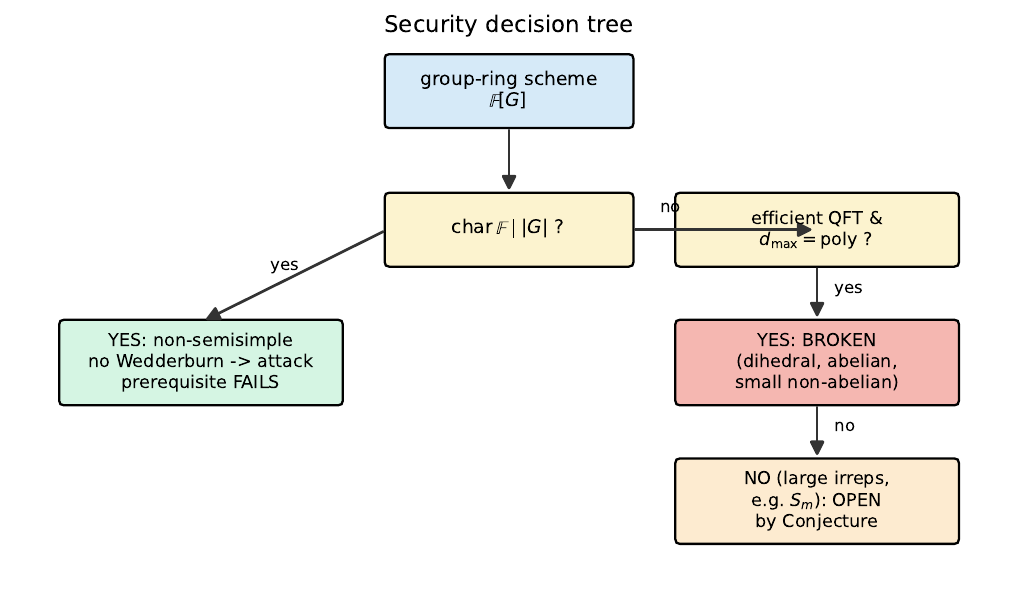}
\caption{A decision procedure for a candidate group-ring platform. Efficient inversion follows when
the algebra is semisimple, has an efficient generalized transform, and has bounded block dimension.
The modular and large-block branches are not reached by our attack, but, per
Section~\ref{sec:surviving}, this is not in itself a proof of security.}
\label{fig:decision}
\end{figure}

\paragraph{Actionable guidance} Three recommendations follow, each qualified by the assumptions
above. First, a group-ring scheme whose security is the hardness of inverting a unit should not be
considered post-quantum on a semisimple, bounded-dimension platform, dihedral included. Second,
choosing a platform for ``non-abelianness'' alone is not protective; what matters is the block
structure and the transform, not commutativity. Third, the only regimes our attack does not reach are
the modular ($\mathrm{char}\,\F\mid|G|$) and large-block regimes, and, as
Sections~\ref{sec:dichotomy} and~\ref{sec:surviving} make clear, neither is established as secure;
they are directions for further study, not safe defaults.

\section{A Construction in the Surviving Regime}\label{sec:surviving}

The dichotomy leaves two regimes that our attack does not reach: super-polynomial representation
dimension and the modular regime $\mathrm{char}\,\F\mid|G|$. We stress that ``not reached by our attack'' is not the same as ``secure.'' Both regimes are
candidate safe harbours only, and the construction below should be read as a proposal for
cryptanalytic study rather than as a scheme we claim to be secure. We treat the modular regime as a design space, and we begin by clearing up a
natural but incorrect inference, namely that non-semisimplicity already implies hardness.

\begin{proposition}[Local modular rings invert easily]\label{prop:localeasy}
For a prime $p$, the ring $\F_p[C_p]$ is non-semisimple (its Jacobson radical $J=(\sigma-1)$ is
nonzero, with $J^p=0$), yet every unit is invertible in $O(\log p)$ Newton iterations, each costing
$O(p)$ field operations.
\end{proposition}
\begin{proof}
In characteristic $p$ one has $\F_p[C_p]\cong\F_p[x]/((x-1)^p)$, a local ring with maximal ideal
$(x-1)$. Writing $u=a+\nu$ with $\nu\in J$ nilpotent, the Newton iteration $v\leftarrow v(2e-uv)$
converges quadratically in $\lceil\log_2 p\rceil$ steps because $e-uv$ lies in the nilpotent radical.
The behaviour is confirmed for $p\in\{3,5,7,11\}$ in Supplementary~\textsection S8.
\end{proof}

The lesson is that the modular escape of Section~\ref{sec:boundary} only removes the Wedderburn
route. On its own it does not manufacture a hard inversion problem, and a secure modular construction
must therefore place its hardness somewhere other than the unit group of a local ring. With this in
mind, take $p$ prime, $\ell$ coprime to $p$, and $R=\F_p[C_p\times C_\ell]$. Then
$R/J\cong\F_p[C_\ell]$ is semisimple and falls to the present method, while the radical
$J=(\sigma-1)R$ has $\F_p$-dimension $(p-1)\ell$. The idea is to sacrifice the semisimple quotient
deliberately and concentrate the secret in the radical.

\begin{definition}[Radical Recovery, RR]\label{def:rr}
Given the semisimple image $\pi(u)\in R/J$ of a secret unit $u$, together with $t<\dim_{\F_p}J$
public linear probes $A\mathbf{r}$ of its radical coordinate vector $\mathbf{r}\in\F_p^{\dim J}$
(with $A\in\F_p^{t\times\dim J}$ public), recover $\mathbf{r}$.
\end{definition}

\begin{theorem}[One-wayness reduces to RR]\label{thm:rr}
Define $f(u)=(\pi(u),A\mathbf{r})$ on units of $R$ with secret radical part $\mathbf{r}$. Any
algorithm inverting $f$ solves RR, and conversely. Moreover the Coherent Wedderburn Inversion
method, applied to $f(u)$, recovers only $\pi(u)$ and yields no information about $\mathbf{r}$ beyond
the public probes: when $t<\dim_{\F_p}J$ the consistent secrets form an affine space of size
$p^{\dim J-\mathrm{rank}\,A}>1$.
\end{theorem}
\begin{proof}
The two components of $f$ are exactly the RR instance, so inverting $f$ and solving RR are
interreducible. The present method operates in $R/J$ and returns $\pi(u)^{-1}$; the radical
coordinates enter $f$ only through $A\mathbf{r}$, so the attacker's view is the linear system
$A\mathbf{x}=A\mathbf{r}$, whose solution set is an affine translate of $\ker A$ of dimension
$\dim J-\mathrm{rank}\,A$. This free dimension is positive on every tested instance
(Supplementary~\textsection S8).
\end{proof}

\begin{remark}[Status of the construction]\label{rem:rr}
Theorem~\ref{thm:rr} is a reduction, not a hardness proof: it shows the construction is one-way
exactly as hard as RR and that the present method does not break it. We do not claim that RR is hard.
It is a new assumption, structurally a noisy-linear-recovery problem masked by the radical, offered
for cryptanalytic study, and the construction should be read as a starting point rather than a
finished scheme.
\end{remark}

\section{Experimental Validation}\label{sec:experiments}

We accompany the paper with an open-source artifact (Python~3.12, Qiskit~2.x) that implements the
validated claims. It provides exact $\F_p$ arithmetic; three inversion routes (abelian character
diagonalization, integral CRT lifting, and the general regular representation); the dihedral
Wedderburn method; the non-abelian Fourier transform; the trace-form semisimplicity witness; resource
estimation; and a simulator-validated, hardware-compatible circuit. A single driver regenerates every
number, figure, and table from fixed random seeds, with the environment pinned in
\texttt{requirements.txt}. This section reports four representative results. The full campaign, namely
the scalability, comparative, statistical, ablation, non-splitting, and non-cyclic studies, is
collected in Supplementary~\textsection S8. The largest directly benchmarked group is $D_{51}$
(order $102$); larger entries in the resource tables are model-based extrapolations, labelled as
such.

Correctness is the first thing to establish. Table~\ref{tab:validation} reports exact validation on
random units drawn from abelian, dihedral, and symmetric platforms: every tested unit is inverted
correctly, and the semisimplicity flag agrees with Theorem~\ref{thm:boundary} in every case.

\begin{table}[t]\centering
\caption{Validation of unit inversion on random units, with the two-sided inverse verified over
$\F_p$.}
\label{tab:validation}
\begin{tabular}{lccccc}
\toprule
Family & $\dim$ & $\mathbb{F}_p$ & units tested & verified & semisimple \\
\midrule
$C_4\times C_2$ (abelian) & 8 & $\mathbb{F}_{17}$ & 189 & 189/189 & yes \\
$D_4$ & 8 & $\mathbb{F}_{5}$ & 92 & 92/92 & yes \\
$D_6$ & 12 & $\mathbb{F}_{13}$ & 191 & 191/191 & yes \\
$S_3$ & 6 & $\mathbb{F}_{7}$ & 171 & 171/171 & yes \\
\bottomrule
\end{tabular}
\end{table}

The asymptotic separation of Theorem~\ref{thm:reduction} is visible already at modest sizes.
Figure~\ref{fig:cost} compares Wedderburn block inversion against dense regular-representation
inversion over $\F_p$ on identical units. The two curves cross near order $14$ and the gap widens
monotonically thereafter, the expected signature of replacing an $|G|^3$ step by an $O(|G|)$ one
together with a fixed transform overhead, rather than a uniform speedup. Relative to the abelian
algorithm of~\cite{dooms2025}, our method is a strict generalization that reduces to theirs when
every block is $1\times1$.

The block dimension is the parameter that ultimately governs tractability, and
Figure~\ref{fig:block}(b) isolates it: holding the field fixed and increasing a synthetic block
dimension $d$, the block-inversion time tracks the predicted $d^\omega$ trend, reaching $37$~ms at
$d=64$. This is the concrete face of the super-polynomial barrier that protects large-block
platforms. Figure~\ref{fig:block}(a) shows that growth in the coefficient-field bit length is only
sub-quadratic, as expected for schoolbook modular arithmetic.

Reliability across many random inputs is summarized in Figure~\ref{fig:stat}. Over $3000$ trials per
family on $D_4$, $D_6$, $D_7$, and $S_3$, every invertible unit encountered was inverted exactly. The
observed success rate is $1.000$, with Wilson $95\%$ lower bounds of at least $0.996$, and the
per-inversion runtimes are tight, reflecting the data-independent control flow of the algorithm.
Finally, Figure~\ref{fig:resource} plots the logical-qubit and dominant-gate model for the
finite-case attack; transpiled QFT counts in a Clifford$+T$ basis support the model at demonstrable
sizes.

\begin{figure}[t]\centering
\begin{minipage}{0.49\linewidth}\centering
\includegraphics[width=\linewidth]{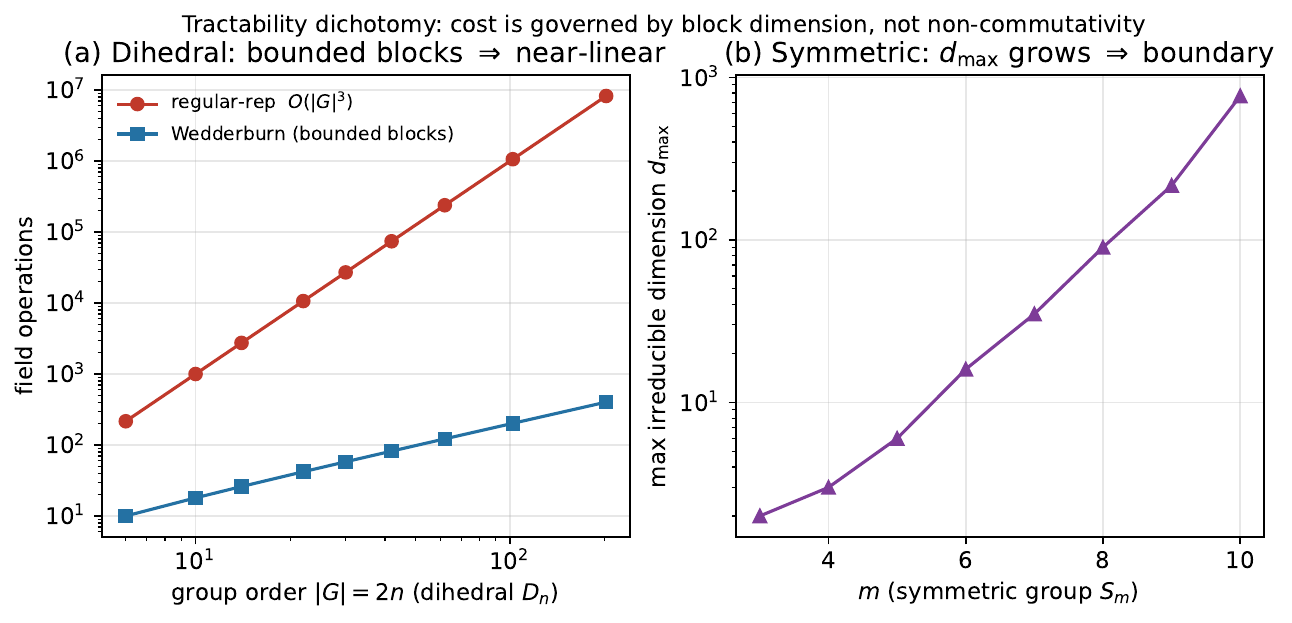}
\caption{Measured runtime of three inverters. The Wedderburn curve crosses below dense $\bmod p$
inversion near order $14$, and the gap widens with $|G|$.}
\label{fig:cost}
\end{minipage}\hfill
\begin{minipage}{0.49\linewidth}\centering
\includegraphics[width=\linewidth]{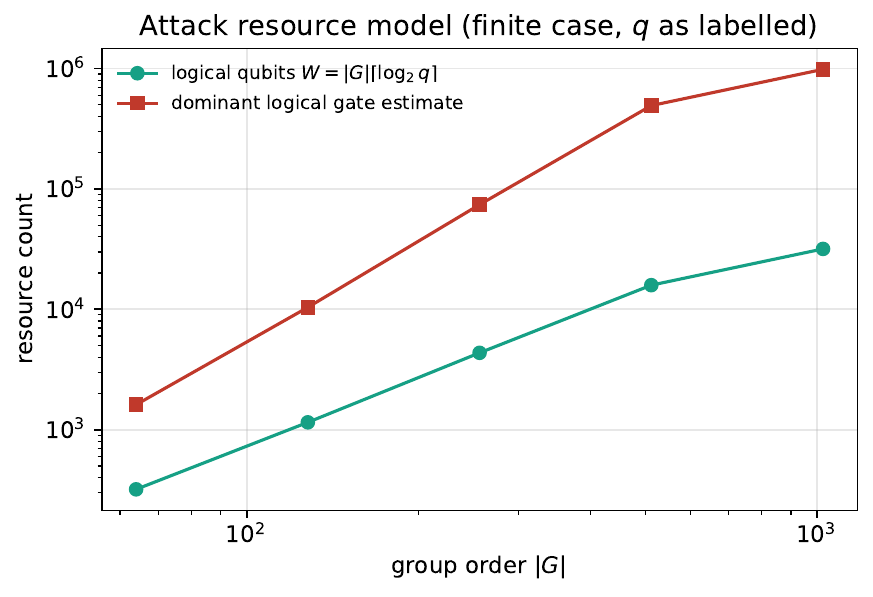}
\caption{Logical resource scaling of the finite-case attack (qubits $W=|G|\lceil\log_2 q\rceil$ and
dominant logical gate estimate).}
\label{fig:resource}
\end{minipage}
\end{figure}

\begin{figure}[t]\centering
\begin{minipage}{0.49\linewidth}\centering
\includegraphics[width=\linewidth]{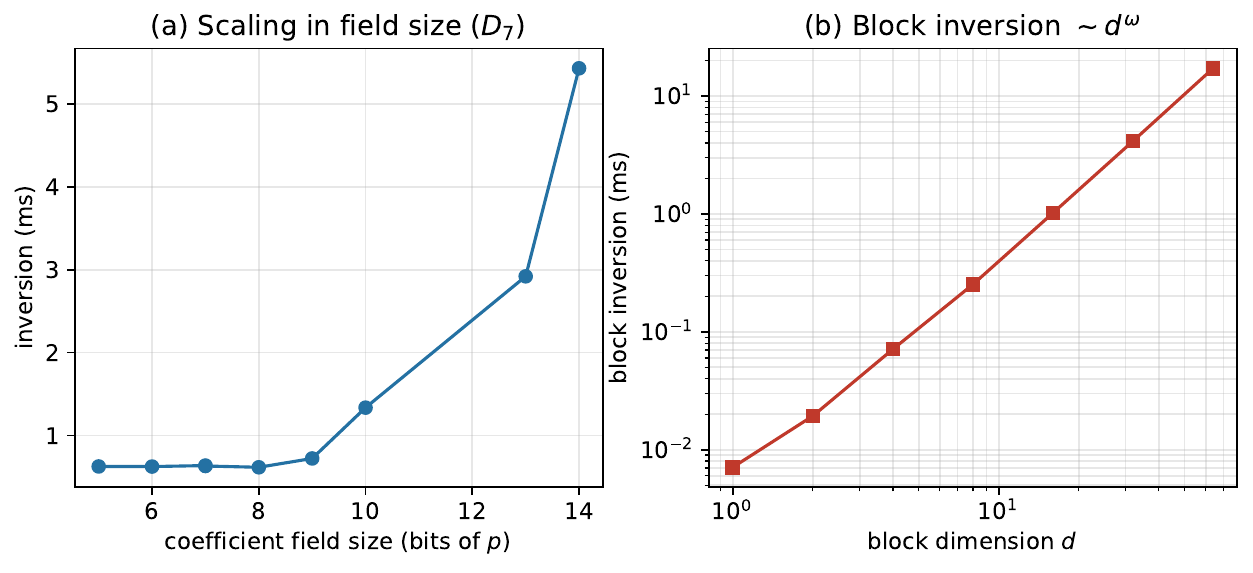}
\caption{(a) Inversion time versus coefficient-field bit length ($D_7$). (b) Single-block inversion
time versus block dimension $d$, tracking $d^\omega$.}
\label{fig:block}
\end{minipage}\hfill
\begin{minipage}{0.49\linewidth}\centering
\includegraphics[width=\linewidth]{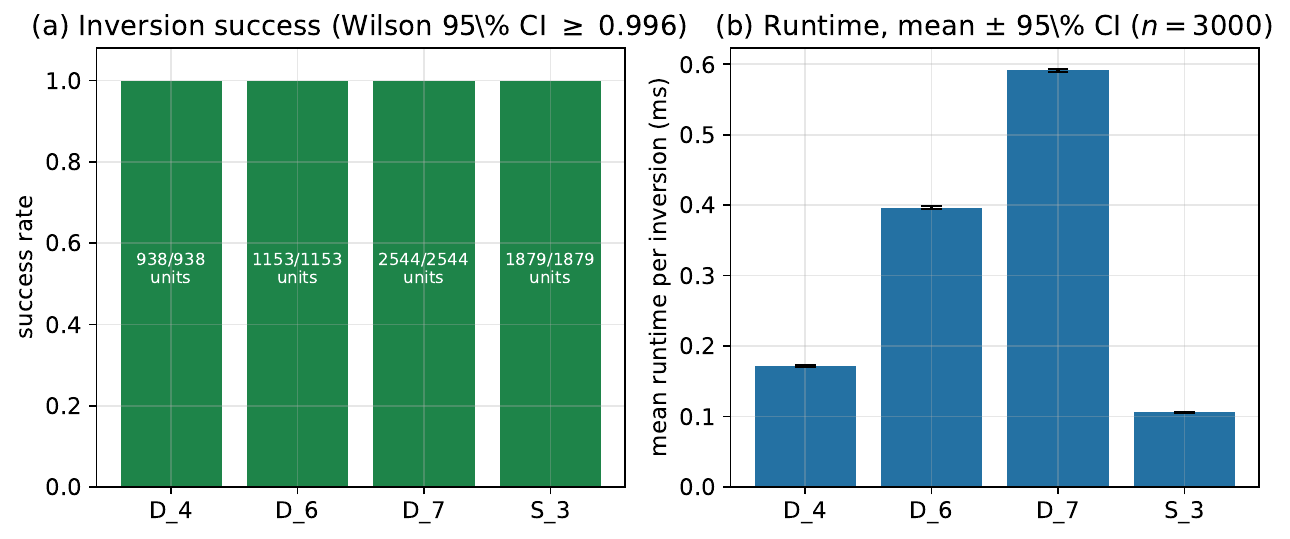}
\caption{(a) Inversion success rate with Wilson $95\%$ confidence; all units invert exactly. (b)
Mean per-inversion runtime with $95\%$ confidence intervals ($n=3000$).}
\label{fig:stat}
\end{minipage}
\end{figure}

\section{Quantum Resource Estimates}\label{sec:resources}

A group-ring element of $\F_q[G]$ occupies a register of width $W=|G|\lceil\log_2 q\rceil$ qubits,
and the attack applies one QFT over $G$ followed by $|G|$ blockwise inversions in $\F_q$. We
summarize the logical model here and defer the details to Supplementary~\textsection S9. Those
details include the full resource tables, the total $T$-count breakdown, and a surface-code estimate.
The surface-code estimate covers magic-state distillation and spacetime volume under a rotated
surface code at physical error $10^{-3}$ with eight $15$-to-$1$ factories. To give a sense of scale,
take two endpoints. The smallest illustrative instance, $D_{31}$, requires code distance $11$, about
$1.2\times10^5$ physical qubits, and sub-second runtime. The largest tabulated instance, $D_{1024}$,
requires about $2\times10^7$ physical qubits, roughly an order of magnitude below contemporary
estimates for breaking RSA-2048~\cite{gidney2021}. These are model-based logical-to-physical
projections, and we treat them as such. A simulator-validated, IBM heavy-hex--compatible
demonstration of the attack's core diagonal step (ideal success $1.0$, degrading gracefully with
transpiled depth) is described in Supplementary~\textsection S10. A ready-to-run hardware-submission
script accompanies the artifact, and on-device execution is left as a deployment step.

\section{Related Work}\label{sec:related}

The work sits at the meeting point of two literatures that had not been connected: the computational
algebra of semisimple decompositions, which knew how to block-diagonalize but did not ask about
quantum cryptographic hardness, and the security analysis of group-ring schemes, which reasoned about
the hidden subgroup problem but did not use the decomposition. Our contribution is the bridge, and the
observation that crossing it removes the assumed hardness rather than supplying it.
Table~\ref{tab:gap} states, for each of the five strands below, the gap it left open and how the
present paper closes it.

\begin{table}[t]\centering
\caption{The five strands of prior work, the gap each left open, and how this paper closes it.}
\label{tab:gap}
\footnotesize\setlength{\tabcolsep}{3pt}
\renewcommand{\arraystretch}{1.25}
\begin{tabular}{p{2.87cm}p{3.61cm}p{5.08cm}}
\toprule
Prior strand & Gap left open & How this paper closes it \\
\midrule
Dooms--Emerencia~\cite{dooms2025} (abelian inversion) & Non-abelian case; cost; affected schemes & Settles bounded-dimension non-abelian case; identifies $d_{\max}$; gives circuit, resources, scheme map \\
Roman'kov--Ushakov~\cite{romankov2023} (algebra-scheme breaks) & Why those breaks worked; which class shares the weakness & Subsumes them under the semisimple-algebra principle (Thm.~\ref{thm:algebra}) \\
Wedderburn algorithms~\cite{friedlronyai,bremner2011} & Whether a quantum trapdoor survives & Supplies the cryptanalytic reading and the coherent quantum realization \\
Non-abelian QFT~\cite{beals1997,moore2006,dihedralgauge2024} & That the transform alone breaks inversion trapdoors & Shows the transform \emph{is} the attack; reuses the circuits as cryptanalysis \\
Dihedral-HSP~\cite{ettinger1999,kuperberg2005,regev2004} & Whether HSP hardness protects group-ring schemes & Proves it does not---the HSP is never invoked (\textsection\ref{sec:hspsep}) \\
\bottomrule
\end{tabular}
\end{table}

The immediate predecessor is the work of Dooms and Emerencia~\cite{dooms2025}, who break unit
inversion for finite commutative coefficient rings and abelian groups using the abelian QFT in a
superposition-and-measure construction. Our result is a strict generalization along the only axis
that mattered for their hardness claim, namely commutativity. Theorem~\ref{thm:gui} recovers their
result as the case $d_{\max}=1$ and extends it to every semisimple platform with an efficient
transform and bounded blocks, while Theorem~\ref{thm:boundary} explains, through the trace form,
exactly where the method must stop. Where their paper leaves the non-abelian question open and offers
neither an implementation nor a resource analysis, we settle the non-abelian case for bounded-block
groups, provide a validated artifact, and quantify resources. Table~\ref{tab:scope} contrasts the
two works along these axes.

\begin{table}[t]\centering
\caption{Capability comparison with the immediate predecessor. ``Explicit circuit'' denotes a stated
reversible register layout with gate and $T$-count accounting; ``resource estimate'' denotes logical
and fault-tolerant qubit and gate figures.}
\label{tab:scope}
\begin{tabular}{lcccc}
\toprule
Work & Abelian & Non-abelian & Explicit circuit & Resource estimate \\
\midrule
Dooms--Emerencia~\cite{dooms2025} & \checkmark & --- & --- & --- \\
This work & \checkmark & \checkmark & \checkmark & \checkmark \\
\bottomrule
\end{tabular}
\end{table}

A second strand is the algebraic and quantum cryptanalysis of algebra-based schemes. Roman'kov and
Ushakov~\cite{romankov2023} broke two associative-algebra signature schemes by a combination of
algebraic and quantum techniques. Theorem~\ref{thm:algebra} subsumes such breaks under a single
principle: any finite-dimensional semisimple algebra with an efficient generalized Fourier transform
and bounded blocks has easy unit inversion. The classical algorithms for computing Wedderburn
decompositions~\cite{bremner2011,friedlronyai} are the constructive backbone we invoke. More broadly,
our methodological stance is that identifying the right algebraic invariant can collapse a
presumed-hard problem. This stance is shared with recent post-quantum cryptanalysis, of which the
SIDH key-recovery attacks~\cite{castryck2022} are a prominent example; in our setting the trace form
plays that role.

The third strand is the dihedral hidden subgroup problem itself. It has been studied through the
measurement analysis of Ettinger and H\o yer~\cite{ettinger1999}, Kuperberg's subexponential
sieve~\cite{kuperberg2005}, and Regev's reduction tying dihedral-HSP hardness to the shortest-vector
problem~\cite{regev2004}; Horan and Kahrobaei~\cite{horan2018} survey its role in group-based
post-quantum cryptography. Our contact with this literature is deliberately negative: the hardness it
establishes is irrelevant to unit inversion, because inversion never forms or sieves a coset state.
The dihedral group is the clean example separating ``efficient QFT'' from ``efficient HSP.''

The fourth strand comprises classical attacks on dihedral group-ring schemes. Tinani~\cite{tinani2022}
breaks a scheme over twisted dihedral group algebras by reducing a decomposition problem to
circulant equations, and the cryptanalysis of dihedral Group-Ring NTRU~\cite{grntru2025} attacks the
underlying shortest-vector problem with lattice reduction. Neither concerns unit inversion and
neither is quantum; our method is orthogonal to both, so that the dihedral group-ring design space is
now constrained from both the classical-lattice and the quantum-algebraic sides.

The fifth strand is the construction of efficient quantum Fourier transforms over non-abelian
groups. This line was initiated by Beals~\cite{beals1997} and generalized by Moore, Rockmore and
Russell~\cite{moore2006}. The specific dihedral circuits we use appear in the quantum-simulation
literature for lattice gauge theories~\cite{dihedralgauge2024}, where the dihedral transform, group
multiplication, and inversion are explicit primitives. The cryptanalytic and simulation communities
thus rely on the same circuits, a coincidence that makes our circuits directly reusable as simulation
primitives. Table~\ref{tab:related} summarizes the comparison across all of these.

\begin{table}[t]\centering
\caption{Comparison with prior work. The columns ``HSP?'' and ``QFT?'' indicate reliance on solving
a hidden subgroup problem and on a quantum Fourier transform, respectively.}
\label{tab:related}
\footnotesize\setlength{\tabcolsep}{3pt}
\begin{tabular}{p{2.13cm}p{1.89cm}p{1.64cm}ccp{1.80cm}p{1.97cm}}\toprule
Paper & Target structure & Attack type & HSP? & QFT? & Quantum complexity & Main limitation \\\midrule

Dooms--Emerencia \cite{dooms2025} & abelian group ring & unit inversion & no & yes (abelian) & $\widetilde{O}(\log|G|)$ & commutative only \\
This work & $\mathbb{F}[G]$, semisimple alg. & unit inversion & no & yes (non-abelian) & $\mathrm{poly}$ if $d_{\max}{=}\mathrm{poly}$ & splitting field; logical cost \\
Roman'kov--Ushakov \cite{romankov2023} & associative algebra sig. & algebraic + quantum & no & partial & scheme-specific & specific schemes \\
Tinani \cite{tinani2022} & twisted dihedral alg. & classical decomposition & no & no & classical & classical only \\
GR-NTRU dihedral \cite{grntru2025} & dihedral GR-NTRU & lattice / SVP & no & no & classical (lattice) & lattice-reduction based \\
Kuperberg \cite{kuperberg2005} & dihedral group & HSP (sieve) & yes & yes & $2^{O(\sqrt{\log})}$ & solves HSP, not inversion \\
Regev \cite{regev2004} & dihedral group & HSP $\leftrightarrow$ SVP & yes & yes & subexponential & reduction, not attack \\
\bottomrule
\end{tabular}
\end{table}

\section{Limitations}\label{sec:threats}

Several assumptions bound the strength of our conclusions, and we state them plainly. The quantum
efficiency claim of Theorems~\ref{thm:gui}(iii) and~\ref{thm:reduction}(ii) presupposes an efficient
generalized QFT for the platform. This holds for the abelian, dihedral, and supersolvable
groups~\cite{beals1997,moore2006,dihedralgauge2024} that cover the schemes we target, but it is not
known in general, and for an arbitrary group the cryptanalytic conclusion does not follow from our
results alone. A related convenience is the splitting-field hypothesis under which
Theorem~\ref{thm:gui} is stated. Over non-splitting finite fields the blocks become matrix algebras
over field extensions (division rings over $\F_p$ being fields by Wedderburn's little theorem). The
base-field regular-representation route still inverts every unit exactly (validated in
Supplementary~\textsection S8), but the clean block presentation requires the extension. Efficiency
also requires bounded blocks, $d_{\max}=\poly(\log|G|)$; this fails for $S_m$, and that regime is open
by Conjecture~\ref{conj:hard}.

On the engineering side, all fault-tolerant figures in Section~\ref{sec:resources} and
Supplementary~\textsection S9 are logical-to-physical projections under a standard surface-code
model rather than hardware measurements, and no on-device execution was performed. On the
complexity-theoretic side, we prove only the easy direction without conditions; the hard direction
rests on oracle-model, output-size, and conditional results, while the unconditional general lower
bound (Conjecture~\ref{conj:hard}) remains open. Finally, because the targeted schemes fix no
standardized parameter sets and, for~\cite{hurley2011,mittalkumar2021}, provide no formal reduction,
our end-to-end recovery is against the intended hardness assumption rather than against a published
security theorem. None of these caveats affects the constructive results, which are validated
exactly.

\section{Conclusion}\label{sec:conclusion}

The lesson of this paper is not that dihedral group rings are weak. It is that non-commutativity was
never the right place to look for hardness. Inverting a unit of a group ring is a
Fourier-diagonalizable task, not a hidden-subgroup task, and the quantity that decides whether it
resists a quantum adversary is the representation-theoretic block dimension of the algebra, together
with the efficiency of its transform---not whether the group is abelian. The hidden subgroup problem,
invoked to justify the move to non-abelian platforms, is never engaged by the attack, so its
hardness, however firm, protects nothing here.

We prove this in one direction---an efficient inversion algorithm, the exact conditions under which
it runs, and an exact witness of the boundary where it stops---and we support the converse with
output-size, query-complexity, and \textsc{\#P}-hardness evidence, leaving a single general lower
bound as an explicit conjecture. We make the attack concrete with a validated reversible circuit and
fault-tolerant resource estimates, and we are careful not to overstate either the reach of the attack
or the security of the regimes it does not reach. The practical recommendation follows directly: the
security analysis of any future group-ring proposal should begin with its block structure and the
efficiency of its transform, and should not rest on the hidden subgroup problem.

Several questions follow naturally.

\begin{itemize}[leftmargin=1.4em,itemsep=2pt]
\item Do structured subclasses of large-block groups admit efficient transforms and bounded
effective blocks? An answer would sharpen Conjecture~\ref{conj:hard}.
\item Can the finite-field result be lifted to $\Z[G]$ for non-abelian $G$ through coefficient-growth
bounds?
\item What is the unit-group structure of modular group algebras from a cryptographic standpoint,
and is the Radical Recovery assumption hard?
\item Can the engineering be completed with a fully fault-tolerant non-abelian phase step and an
end-to-end recovery against specific published parameter sets?
\end{itemize}


\section*{Declarations}

\paragraph{Funding} This research did not receive any specific grant from funding agencies in the
public, commercial, or not-for-profit sectors.

\paragraph{Competing interests} The author declares that he has no known competing financial
interests or personal relationships that could have appeared to influence the work reported in this
paper.

\paragraph{Data availability} An environment-pinned artifact reproducing every theorem check, figure,
and table accompanies this submission as supplementary material. No external datasets were generated
or analysed.

\paragraph{Code availability} The software artifact that regenerates every constructive claim,
figure, and table is provided as supplementary material and is environment-pinned for
reproducibility.

\paragraph{Author contributions} B.~Gupta carried out the conceptualization, methodology, software,
formal analysis, validation, investigation, and writing of this paper in its entirety.

\paragraph{Use of generative AI} During the preparation of this work the author used a generative AI
assistant for language editing, restructuring, and formatting. After using this tool the author
reviewed and edited the content as required and takes full responsibility for the content of the
publication.

\bibliographystyle{plainnat}
\bibliography{refs}

@article{dooms2025,
  author  = {Dooms, A. and Emerencia, C.},
  title   = {Efficient quantum algorithms to break group ring cryptosystems},
  journal = {Journal of Information Security and Applications},
  volume  = {88},
  pages   = {103923},
  year    = {2025}
}

@article{hurley2011,
  author  = {Hurley, B. and Hurley, T.},
  title   = {Group ring cryptography},
  journal = {International Journal of Pure and Applied Mathematics},
  volume  = {69},
  pages   = {67--86},
  year    = {2011}
}

@article{mittalkumar2021,
  author  = {Mittal, G. and Kumar, S. and Narain, S. and Kumar, S.},
  title   = {Group rings based public key cryptosystems},
  journal = {Journal of Discrete Mathematical Sciences and Cryptography},
  year    = {2021}
}

@article{idbased2022,
  author  = {Anonymous},
  title   = {A quantum secure ID-based cryptographic encryption based on group rings},
  journal = {S\=adhan\=a},
  volume  = {47},
  pages   = {Article 35},
  year    = {2022}
}

@article{romankov2023,
  author  = {Roman'kov, V. and Ushakov, A.},
  title   = {Algebraic and quantum attacks on two digital signature schemes},
  journal = {Journal of Mathematical Cryptology},
  year    = {2023}
}

@inproceedings{beals1997,
  author    = {Beals, R.},
  title     = {Quantum computation of Fourier transforms over symmetric groups},
  booktitle = {Proceedings of STOC},
  pages     = {48--53},
  year      = {1997}
}

@article{moore2006,
  author  = {Moore, C. and Rockmore, D. and Russell, A.},
  title   = {Generic quantum Fourier transforms},
  journal = {ACM Transactions on Algorithms},
  volume  = {2},
  pages   = {707--723},
  year    = {2006}
}

@article{kuperberg2005,
  author  = {Kuperberg, G.},
  title   = {A subexponential-time quantum algorithm for the dihedral hidden subgroup problem},
  journal = {SIAM Journal on Computing},
  volume  = {35},
  pages   = {170--188},
  year    = {2005}
}

@article{regev2004,
  author  = {Regev, O.},
  title   = {Quantum computation and lattice problems},
  journal = {SIAM Journal on Computing},
  volume  = {33},
  pages   = {738--760},
  year    = {2004}
}

@article{dihedralgauge2024,
  author  = {Gustafson, E. J. and others},
  title   = {Primitive quantum gates for dihedral gauge theories; Highly-efficient quantum Fourier transformations for certain non-abelian groups},
  journal = {Physical Review D},
  volume  = {110},
  pages   = {074501},
  year    = {2024}
}

@article{bremner2011,
  author  = {Bremner, M. R.},
  title   = {How to compute the Wedderburn decomposition of a finite-dimensional associative algebra},
  journal = {Groups Complexity Cryptology},
  volume  = {3},
  pages   = {47--66},
  year    = {2011}
}

@inproceedings{friedlronyai,
  author    = {Friedl, K. and R\'onyai, L.},
  title     = {Polynomial time solutions of some problems in computational algebra},
  booktitle = {Proceedings of STOC},
  pages     = {153--162},
  year      = {1985}
}

@misc{tinani2022,
  author = {Tinani, S.},
  title  = {Cryptanalysis of a system based on twisted dihedral group algebras},
  note   = {Preprint},
  year   = {2022}
}

@article{grntru2025,
  author  = {Kumar, V. and others},
  title   = {Cryptanalysis of Group Ring NTRU: the case of the dihedral group},
  journal = {Security and Privacy},
  year    = {2025}
}

@inproceedings{castryck2022,
  author    = {Castryck, W. and Decru, T.},
  title     = {An efficient key recovery attack on SIDH},
  booktitle = {EUROCRYPT},
  year      = {2023}
}

@inproceedings{horan2018,
  author    = {Horan, K. and Kahrobaei, D.},
  title     = {The hidden subgroup problem and post-quantum group-based cryptography},
  booktitle = {ICMS, LNCS 10931},
  year      = {2018}
}

@inproceedings{roetteler2017,
  author    = {Roetteler, M. and Naehrig, M. and Svore, K. M. and Lauter, K.},
  title     = {Quantum resource estimates for computing elliptic curve discrete logarithms},
  booktitle = {ASIACRYPT, LNCS 10625},
  pages     = {241--270},
  year      = {2017}
}

@article{gidney2021,
  author  = {Gidney, C. and Eker\aa, M.},
  title   = {How to factor 2048-bit RSA integers in 8 hours using 20 million noisy qubits},
  journal = {Quantum},
  volume  = {5},
  pages   = {433},
  year    = {2021}
}

@inproceedings{williams2024,
  author    = {Williams, V. V. and Xu, Y. and Xu, Z. and Zhou, R.},
  title     = {New bounds for matrix multiplication: from alpha to omega},
  booktitle = {SODA},
  pages     = {3792--3835},
  year      = {2024}
}

@inproceedings{ettinger1999,
  author    = {Ettinger, M. and H\o yer, P.},
  title     = {On quantum algorithms for noncommutative hidden subgroups},
  booktitle = {STACS, LNCS 1563},
  pages     = {478--487},
  year      = {1999}
}

@article{burgisser2000,
  author  = {B\"urgisser, P.},
  title   = {The computational complexity of immanants},
  journal = {SIAM Journal on Computing},
  volume  = {30},
  number  = {3},
  pages   = {1023--1040},
  year    = {2000}
}

@article{delacruz2021,
  author       = {de la Cruz, Javier and Villanueva-Polanco, Ricardo},
  title        = {Public key cryptography based on twisted dihedral group algebras},
  journal      = {arXiv preprint arXiv:2112.07798},
  year         = {2021},
  note         = {Adv. Math. Commun.}
}

\end{document}